\documentclass[letterpaper,11pt]{article}

\usepackage[margin=0.9in]{geometry}
\usepackage{amsmath,amssymb}
\usepackage{amsthm}
\usepackage{algorithm}
\usepackage{algpseudocode}
\usepackage{graphicx}
\usepackage{textcomp}
\usepackage{xcolor}
\usepackage{changepage}
\usepackage{url}
\usepackage{hyperref}
\usepackage[capitalize]{cleveref}
\usepackage{fancyhdr}
\usepackage{booktabs}

\usepackage{tikz}
\usepackage[shortlabels]{enumitem}
\usepackage{thm-restate}
\usepackage{float}        
\usepackage{caption}      
\usepackage{subcaption}
\usepackage[numbers]{natbib}
\usepackage{lineno}
\usepackage[T1]{fontenc} 
\usepackage{titlesec}

\titleformat{\subsection}[runin] 
  {\normalfont\bfseries}         
  {\thesubsection}               
  {1em}                          
  {}                             
  [.]

\titlespacing*{\section}
  {0pt}      
  {2ex plus 1ex minus .2ex} 
  {1ex plus .2ex}           
  
\newcommand{\cA}{\mathcal{A}}

\newcommand{\bG}{\mathbb{G}}

\newcommand{\ora}[1]{\overrightarrow{#1}}

\DeclareMathOperator{\worstcase}{Worst-Case}
\DeclareMathOperator{\universal}{Universal}
\DeclareMathOperator{\daguniversal}{Order}
\DeclareMathOperator{\Root}{root}
\DeclareMathOperator{\runtime}{Running-Time}

\AtBeginDocument{}

\theoremstyle{definition}
\newtheorem{definition}{Definition}

\theoremstyle{plain}

\newtheorem{claim}{Claim}

\newtheorem{proposition}{Proposition}
\newtheorem{theorem}{Theorem}
\newtheorem{lemma}{Lemma}
\theoremstyle{remark}

\Crefname{claim}{Claim}{Claims}

\title{Tight Better-Than-Worst-Case Bounds for Element Distinctness and Set Intersection}

\makeatletter
\let\@fnsymbol\@arabic
\makeatother

\author{%
  Ivor van der Hoog\thanks{Department of Theoretical Computer Science, IT University of Copenhagen, Denmark}
  \and
  Eva Rotenberg\textsuperscript{1}
  \and
  Daniel Rutschmann\textsuperscript{1}
}
\date{} 

\begin{document}
\title{ Tight Better-Than-Worst-Case Bounds for Element Distinctness and Set Intersection}


\maketitle

\begin{abstract}
\begin{adjustwidth}{-1cm}{-1cm}
The \emph{element distinctness} problem takes as input a list~$I$ of~$n$ values from a totally ordered universe, where pairwise comparisons between values are allowed, and the goal is to decide whether~$I$ contains any duplicates. 
It is a well-studied problem with a classical worst-case~$\Omega(n \log n)$ comparison-based lower bound by Fredman [TCS'76]. 
At first glance, this lower bound appears to rule out any algorithm more efficient than the naive approach of sorting~$I$ and comparing adjacent elements. 
However, upon closer inspection, the~$\Omega(n \log n)$ bound is overly pessimistic. 
For instance, if~$I$ consists of~$n/2$ identical elements and~$n/2$ distinct ones, a median-finding algorithm will, regardless of the input order, find a duplicate in linear time. 
This raises a natural question:

\begin{quote}
\emph{Are there comparison-based lower bounds for element distinctness that are sensitive to the amount of duplicates in the input instance?}
\end{quote}

To address this question, we derive instance-specific lower bounds. For any input instance $I$, we represent the combinatorial structure of the duplicates in~$I$ by an undirected graph~$G(I)$ that connects identical elements. 
Each such graph~$G$ is a union of cliques, and we study algorithms by their worst-case running time over all inputs~$I'$ with~$G(I') \cong G$. 
We establish an adversarial lower bound showing that, for any deterministic algorithm~$\mathcal{A}$, there exists a graph~$G$ and an algorithm~$\mathcal{A}'$ that, for all inputs $I$ with $G(I) \cong G$, is a factor~$O(\log \log n)$ faster than~$\mathcal{A}$. 
Consequently, no deterministic algorithm can be~$o(\log \log n)$-competitive for all graphs~$G$. 
We complement this with an~$O(\log \log n)$-competitive deterministic algorithm, thereby obtaining tight bounds for element distinctness that go beyond classical worst-case analysis.
Subsequently, we 
study the related problem of \emph{set intersection}. 
We show that no deterministic set intersection algorithm can be $o(\log n)$-competitive, and provide an $O(\log n)$-competitive deterministic algorithm. 
This shows a separation between element distinctness and the set intersection problem.
\end{adjustwidth}
\end{abstract}

\paragraph{Funding.} This work was supported by the the VILLUM Foundation grant (VIL37507) ``Efficient Recomputations for Changeful Problems'' and by the Carlsberg Foundation, grant CF24-1929.

\pagenumbering{gobble}
\newpage

\setcounter{page}{0}





\pagenumbering{arabic}
\section{Introduction}

The \emph{element distinctness} problem takes as input a list $I$ of $n$ values, where pairwise comparisons between values are allowed. The goal is to decide whether $I$ contains any duplicates. 
In the closely related problem of 
\emph{set intersection},  the input consists of two lists $A$ and $B$, each containing $n$ values, and the task is to determine whether there exist $(a, b) \in A \times B$  such that $a = b$.

\paragraph{Complexity.}
The algorithmic complexity of element distinctness depends on the operations permitted.  
The most general model for this problem is the \emph{comparison-based} model, in which one may perform pairwise comparisons between elements of $I$, but cannot otherwise manipulate its values. 
In this model, the worst-case complexity depends on the type of comparisons that are allowed. 
For instance, if a comparison between two elements $a, b \in I$ only reveals whether $a = b$ or $a \neq b$, then element distinctness has a trivial $\Theta(n^2)$ worst-case lower bound. 
The canonical comparison-based model assumes that the input $I$ consists of elements from a totally ordered set. That is, for any $a, b \in I$, either they are duplicates ($a=b$), or they satisfy $a < b$ or $b < a$. 
In this model, Fredman~\cite[TCS'76]{fredman_how_1976} proved an $\Omega(n \log n)$ information-theoretic lower bound by showing that any correct algorithm must produce a certificate verifying that no duplicates exist. 
If $I$ contains no duplicates, such a certificate defines a total order on $I$. 
Since there are $n!$ distinct total orders, any comparison-based algorithm requires $\Omega(n \log n)$ comparisons. 
This bound is tight: sorting $I$ and checking adjacent elements suffices to decide distinctness.
However, this lower bound is rather coarse, and later work aimed to obtain more refined bounds. 
Ajtai~\cite[Combinatorica'88]{Ajtai1988} extended the analysis by incorporating space usage. 
Let $S$ denote the number of available RAM cells. 
Ajtai proved a comparison-based lower bound of $\Omega(f(n)/S)$, where $f(n)$ is superlinear. 
Beame, Borodin, Saks, and Skyum~\cite[STOC'89]{BeameBorodinSaksSkyum1991} refined this to $\Omega(n^{3/2}/\sqrt{S})$. 
Independently, Yao~\cite[FOCS'88]{YaoNear-Optimal} obtained a comparison-based lower bound, conjectured to be near-optimal, of $\Omega\bigl(n^{2 - 1/\sqrt{\log n}} / S\bigr)$. 
Patt-Shamir and Peleg~\cite[TCS'93]{pattshamir1993Space} show an $\Omega(n^{3/2}/\sqrt{S})$ lower bound for set intersection.
Further bounds are known for the \emph{$k$-set intersection} problem, where the number of sets $k$ is part of the input. 
P\u{a}tra\c{s}cu~\cite[STOC'10]{Patrascu} reduced 3SUM to $k$-set intersection, yielding a conditional near-quadratic lower bound. 
Kopelowitz, Pettie, and Porat~\cite[SODA'16]{Kopelowitz2016Higher} strengthened this result by excluding strongly subquadratic algorithms. 
Chakrabarti, Cormode, and McGregor~\cite[STOC'08]{Chakrabarti2008} established lower bounds for $k$-set intersection in the communication and streaming models, showing that $\Omega(n)$ bits of space are required for single-pass streaming algorithms.

\paragraph{Faster algorithms in specialized models of computation.}
Given the $\Omega(n \log n)$ comparison-based lower bound, faster algorithms for element distinctness have been achieved only in less general models of computation. 
Ostlin and Pagh~\cite[STOC'03]{Ostlin2003Hashing} showed that element distinctness can be solved in expected linear time and space when hashing is allowed. 
Beame, Clifford, and Machmouchi~\cite[FOCS'13]{Beame2013FOCS} revisited the space--time trade-off introduced by Ajtai~\cite{Ajtai1988} and showed that, with access to hashing and integer arithmetic, one can achieve a randomized trade-off of $O(n^{3/2}/\sqrt{S})$. 
Other improvements rely on models that permit faster sorting. 
Assuming that $I$ consists of integers and that the underlying word-RAM supports constant-time integer arithmetic, Han~\cite[STOC'02]{Han2002Deterministic} gave a deterministic $O(n \log \log n)$ sorting algorithm, which directly yields an $O(n \log \log n)$ algorithm for element distinctness. 
Han and Thorup~\cite{Han2002Random} further obtained a randomized $O(n \sqrt{\log \log n})$ algorithm for the same problem. 
Finally, P\u{a}tra\c{s}cu and Demaine~\cite[SICOMP'06]{Patrascu} proved an $\Omega(\log n)$ lower bound for dynamic element distinctness and dynamic set intersection in the cell-probe model.

\paragraph{Problem statement.}
At first glance, Fredman’s $\Omega(n \log n)$ lower bound appears to rule out faster comparison-based algorithms. 
Subsequent work has therefore focused on proving lower bounds in other settings or under more specialized computational models. 
However, upon closer examination, the $\Omega(n \log n)$ bound is overly coarse.  
Consider an input instance $I$ in which all elements are identical.  \newpage

In this case, any algorithm can correctly terminate after a single comparison. 
Consider the next example, where half of the elements in $I$ are identical and the remaining half are distinct. 
Here, an adversarial argument yields an $\Omega(n)$ lower bound: any correct algorithm must verify its output by identifying a duplicate that serves as a \emph{witness}, and an adversary can force it to process all $n/2$ distinct elements before encountering any duplicate. 
This linear lower bound can be matched by any algorithm that, for instance, finds the median of $I$, regardless of the order in which the input is presented.
These observations naturally raise the question:
\emph{Are there lower bounds for element distinctness that are sensitive to the amount of duplicates in the input instance?}
To answer this question, we examine several frameworks that refine worst-case analysis by considering algorithmic performance under input-specific parameters. 
We explore a variety of such frameworks, and provide  tight bounds.

\paragraph{Better-than-worst-case optimality.}
There exists a variety of notions for defining better-than-worst-case complexity, and a long history of results achieving such bounds. 
All these approaches share a common theme: if the worst-case running time of an algorithm is defined as the maximum over all possible input instances, can we instead derive lower bounds from the \emph{input instance}?

The strictest interpretation of this idea is \emph{true instance optimality} (sometimes referred to as $O(1)$-competitiveness). 
An algorithm $\cA$ is \emph{instance-optimal} if there exists a constant $c$ such that, for all sufficiently large input instances $I$ and for all algorithms $\cA'$, we have
$
\runtime(\cA, I) \leq c \cdot \runtime(\cA', I).
$
Perhaps the most famous instance-optimality conjecture is due to Sleator and Tarjan~\cite{Sleator1985self}, who proposed that there exists a constant $c$ and a BST algorithm $\cA$ such that, for every sufficiently long sequence $I$ of predecessor queries, the total running time $\runtime(\cA, I)$ of answering all queries in $I$ satisfies
$
\runtime(\cA, I) \leq c \cdot \runtime(\cA', I)
$
for every other BST algorithm $\cA'$. 
This conjecture remains open, but Demaine et al.~\cite[FOCS'04]{Demaine2004Dynamic} showed a BST algorithm $\cA$ that is $O(\log \log n)$-instance-competitive. Formally, they showed that, for any sufficiently long sequence $I$ of BST queries and every other BST $\cA'$, 
$
\runtime(\cA, I) \leq O(\log \log n) \cdot \runtime(\cA', I).
$

True instance optimality is almost always provably unattainable. 
This also holds for element distinctness. 
Consider the family of quadratic-time algorithms $\{ \cA_{i,j} \}$ that perform all pairwise comparisons of elements in $I$ in arbitrary order, where $\cA_{i,j}$ starts by comparing $I[i]$ and $I[j]$. 
For any deterministic algorithm $\cA$ and constant $c$, there exists an input $I$ and pair $(i,j)$ such that $I[i] = I[j]$ is the only duplicate and, after $c$ comparisons, $\cA$ has not yet compared $I[i]$ and $I[j]$. 
However, the algorithm $\cA_{i,j}$ finds the duplicate after one comparison. 
In fact, this construction shows that no $o(n \log n)$-instance-competitive algorithm can exist: for any algorithm, there exists an input $I$ and pair $(i,j)$ such that $I[i] = I[j]$ is the only duplicate, and after $\log (n!) - 1$ comparisons, $\cA$ cannot certify whether $I[i] < I[j]$, $I[i] > I[j]$, or $I[i] = I[j]$. 
Yet, $\cA_{i,j}$ terminates after one comparison. This calls for a notion of optimality that is more fine-grained than worst-case, but coarser than instance optimality.

\paragraph{Output-optimality.}
Fix some algorithmic problem and a value $k$, and denote by $\mathbb{I}_{n, k}$ all inputs of size $n$ that have an output of size $k$. 
For a fixed algorithm $\cA$ and integer pair $(n, k)$, we define the output-sensitive running time as
$
\textnormal{Output}(\cA, n, k) := \max\limits_{I \in \mathbb{I}_{n, k}} \runtime(\cA, I)$.
We say that an algorithm $\cA$ is \emph{output-optimal} if there exists a constant $c$ such that, for all $k$ and sufficiently large $n$, we have, for all $\cA'$,
$
\textnormal{Output}(\cA, n, k) \leq c \cdot \textnormal{Output}(\cA', n, k)$.
A classical output-optimal result is the $O(n \log k)$-time convex hull algorithm by Kirkpatrick and Seidel~\cite[SICOMP'86]{kirkpatrick1986ultimate}.

A related but stronger notion is \emph{instance-optimality in the order-oblivious model}, a term proposed by Afshani, Barbay, and Chan~\cite[FOCS'09]{afshani2009instance}. 
For a fixed input $I$, let $\Pi(I)$ denote all permutations of $I$. 
For an algorithm $\cA$ and input $I$, we define the order-oblivious running time as
\[
\textnormal{Order}(\cA, I) := \max\limits_{I' \in \Pi(I)} \runtime(\cA, I'). \]
\newpage

\noindent
An algorithm $\cA$ is \emph{instance-optimal in the order-oblivious model} if there exists a constant $c$ such that, for all sufficiently large $I$, $
\textnormal{Order}(\cA, I) \leq c \cdot \textnormal{Order}(\cA', I)$, 
for all algorithms $\cA'$. 
They proved that the algorithm in \cite{kirkpatrick1986ultimate} is optimal in this model. 
Note that this notion of optimality implies output-optimality.%

\paragraph{Universal optimality.}
In recent years, the term \emph{universal optimality} has been proposed as an umbrella concept that captures many notions of better-than-worst-case analysis. 
Haeupler, Wajc, and Zuzic~\cite[STOC'21]{Haeupler2021UniversalDistributed} introduced this concept in a distributed setting, observing that the running time of their algorithm $\cA$ depends on two aspects of the input: the graph topology $G$ of the distributed system and the edge weights $w$. 
The worst-case running time then is:
\[
\worstcase(\cA, n) := \max_{n\textnormal{-vertex graphs } G} \, \max_{\textnormal{weightings } w \textnormal{ of } G} \, \runtime(\cA, G, w).
\]
They argued that this double maximum provides an overly pessimistic view of an algorithm’s performance and proposed a more refined analysis. 
Rather than maximizing over both parameters, they fix the first variable and define the \emph{universal running time} as the maximum over the second:  
\[
\textnormal{For a fixed } n\textnormal{-vertex graph } G, \quad  
\operatorname{Universal}(\cA, G) := \max_{\textnormal{weightings } w \textnormal{ of } G} \runtime(\cA, G, w).
\]

For fixed $G$, we define the \emph{universal lower bound} as the smallest universal running time over all algorithms~$\cA'$, and we say that $\cA$ is universally optimal if, for all sufficiently large $G$, it matches this lower bound up to constant factors.

\paragraph{A universal framework.}
Haeupler, Hladík, Rozho\v{n}, Tarjan, and T\v{e}tek~\cite[FOCS'24]{Haeupler2024Dijkstra} extended this notion beyond distributed settings, observing that it can be generalized into a broader framework. 
If running time depends on two variables $X$ and $Y$, such that:
\[
\operatorname{Worst-case}(\cA, n) = \max\limits_{\textnormal{variable } X \textnormal{ of size } n } \quad  \max\limits_{\textnormal{inputs } Y \textnormal{ derived from } X} \runtime(\cA, X, Y), \]

then we can define the universal running time by fixing $X$, and considering the worst-case $Y$:
\[
\textnormal{For a fixed input } X, \quad 
\textnormal{Universal}(\cA, X) := \max_{\textnormal{inputs } Y \textnormal{ derived from } X} \runtime(\cA, X, Y).
\]

An algorithm~$\cA$ is then said to be \emph{universally optimal} if there exists a constant $c$ such that, for all sufficiently large $X$ and all algorithms $\cA'$, $
\textnormal{Universal}(\cA, X) \leq c \cdot \textnormal{Universal}(\cA', X)$.
Universal optimality is challenging to obtain, since the algorithm $\cA$ must be efficient for all $X$, but competes against algorithms $\cA'$ that can be uniquely designed for inputs sharing that same fixed $X$.
An algorithm $\cA$ is \emph{$f(n)$-competitive} if, for sufficiently large $n$, all $X$ of size $n$, and all $\cA'$, we have $
\universal(\cA, X) \le f(n) \cdot \universal(\cA', X)$.
We say that an algorithm is \emph{$O(f(n))$-competitive} if it is $g(n)$-competitive for some $g \in O(f)$. 
A universally optimal algorithm is $O(1)$-competitive.

Haeupler, Hladík, Rozho\v{n}, Tarjan, and T\v{e}tek~\cite[FOCS'24]{Haeupler2024Dijkstra} studied the single-source shortest path (SSSP) problem. 
They observed that the running time of an SSSP algorithm~$\cA$ depends on two aspects of the input: the underlying directed graph~$G$ and its edge weights~$w$. 
Consequently, the worst-case running time can be expressed as a double maximum over both the graph structure and the weighting. 
The universal running time is then defined analogously: 
\[
\textnormal{For a fixed } n\textnormal{-vertex directed graph } G, \quad  
\textnormal{Universal}(\cA, G) := \max_{\textnormal{weightings } w \textnormal{ of } G} \runtime(\cA, G, w).
\]

There is a long line of research on designing universally competitive algorithms for sorting under partial information. 
Here, the input is a partial order~$P$ (a directed, acyclic graph), and the goal is to determine a hidden underlying linear extension~$L$ of~$P$ using the fewest possible comparisons. 
The running time of a partial-sorting algorithm~$\cA$ depends on both the input poset~$P$ and the linear extension~$L$. 
Thus, the universal running time of an algorithm can be defined as \newpage 
\[
\textnormal{For a fixed DAG } P, \quad  
\textnormal{Universal}(\cA, P) := \max_{\textnormal{linear extensions } L \textnormal{ of } P} \runtime(\cA, P, L).
\]

The design of such better-than-worst-case algorithms predates the term \emph{universal optimality}, and was proposed by Kahn and Saks \cite[Order'84]{kahn_balancing_1984}.
There was a steady stream of increasingly competitive algorithms by Kahn and Kim \cite[STOC'92]{kahn_entropy_1992} and later Cardinal et al~\cite[STOC'10]{cardinal_sorting_2010} until Haeupler, Hladík, Iacono, Rozho\v{n}, Tarjan, and T\v{e}tek~\cite[SODA'25]{Haeupler2025Sorting} (and independently van der Hoog and Rutschmann~\cite[FOCS'24]{Hoog2024Tight}) gave $O(1)$-competitive algorithms. 

Even the notion of Afshani, Barbay, and Chan~\cite[FOCS'09]{afshani2009instance} fits into this framework, since running times depends on the input $P$ and on the order $I_P$ in which the input is received. 
The resulting universal running time, defined by fixing the input $P$ and taking the worst-case running time over all orders $I_P$ of $P$, is thereby equivalent to the order-oblivious running time.
Universal optimality thus provides a general framework for expressing instance-sensitive lower bounds.
Since formalizing this framework, it has been successfully applied across a variety of settings~\cite{Hoog2024Tight, Zuzic2022Universally, Haeupler2021UniversalDistributed, Haeupler2024Dijkstra, Haeupler2025Sorting, hladik2025FORC, Hoog2025SimplerSorting, Hoog2025SimplerDAG, Hoog2025Convex}.

\subsection{Better-than-worst-case element distinctness}
We study better-than-worst-case algorithms for element distinctness in the comparison-based model of Fredman~\cite{fredman_how_1976} and Yao~\cite{YaoNear-Optimal}. 
We apply the universal optimality framework to derive lower bounds when the input contains duplicates.

Universal optimality fixes the combinatorial structure of the input. 
For element distinctness, we identify two natural ways to capture this structure. 
Our first proposal is to describe it through the \emph{duplicates} in $I$. 
For any instance $I$, let $G(I)$ be the labeled graph with vertex set $\{1, \ldots, n\}$ where the edge $\{i,j\}$ exists if and only if $I[i] = I[j]$. 
Thus, $G(I)$ is the union of disjoint cliques. 
We can then fix a graph topology $G$ and consider the worst-case running time on inputs that have the same topology. 
Formally, let $\bG_n$ denote the family of all $n$-vertex graphs that are unions of disjoint cliques, and let $\cong$ denote graph isomorphism. Then $
  \worstcase(\cA,n) = \max\limits_{G \in \bG_n}\; \max\limits_{I:\, G(I)\cong G}\; \runtime(\cA,I).$ And, \vspace{-0.1cm} 
  \[
  \text{for fixed } G \in \bG_n, \quad
  \universal(\cA,G) = \max_{I:\, G(I)\cong G}\; \runtime(\cA,I).
\]\vspace{-0.1cm}

This formulation captures the earlier intuition. 
If $I$ contains only one distinct value, then $G(I)$ is a single clique, and the universal lower bound for this $G$ is constant. 
If $I$ contains $n/2$ identical elements and $n/2$ distinct ones, then $G(I)$ consists of a clique of size $n/2$ and $n/2$ isolated vertices, and the universal lower bound for such $G$ is $\Omega(n)$. 
More generally, the universal lower bound for $G \in \bG_n$ can take any value in $[1,\, n \log n]$. 
In our context, we say that an algorithm is $O(f(n))$-universally competitive if there exists a function $g \in O(f)$ such that, for all sufficiently large $n$, for all graphs $G \in \bG_n$, and for all algorithms $\cA'$, we have $
\universal(\cA,G) \leq g(n) \cdot \universal(\cA',G)$.

Our second proposal is to apply universal optimality in a manner which is equivalent to \emph{instance optimality in the order-oblivious model}. 
For any input $I$, we consider the worst-case running time of an algorithm over all permutations $I'$ of $I$. 
Equivalently, this can be expressed in terms of graph topology: every input $I$ induces a directed graph $\overrightarrow{G}(I)$ whose vertices correspond to the elements of $I$, and where there is an arc from $I[i]$ to $I[j]$ if and only if $I[i] < I[j]$. 
Note that for all inputs $I$ of size $n$, the undirected complement of $\overrightarrow{G}(I)$ lies in $\bG_n$. 
Let $\overrightarrow{\bG}_n$ denote the family of all such directed graphs. 
We then have $
  \worstcase(\cA,n) = \max\limits_{\overrightarrow{G} \in \overrightarrow{\bG}_n}\; \max\limits_{I:\, \overrightarrow{G}(I)\cong \overrightarrow{G}}\; \runtime(\cA,I)$, and\vspace{-0.1cm}
  \[
  \text{for fixed } \overrightarrow{G} \in \overrightarrow{\bG}_n, \quad
  \textnormal{Order}(\cA,\overrightarrow{G}) = \max_{I:\, \overrightarrow{G}(I)\cong \overrightarrow{G}}\; \runtime(\cA,I).
\]
To distinguish between our two models, we say that an algorithm is $O(f(n))$-\emph{order-competitive} if there exists a function $g \in O(f)$ such that, for all sufficiently large $n$, for all $\overrightarrow{G} \in \overrightarrow{\bG}_n$, and all algorithms $\cA'$, 
\[
 \textnormal{Order}(\cA,\overrightarrow{G}) \leq g(n) \cdot \textnormal{Order}(\cA',\overrightarrow{G}).
\]
\newpage

\paragraph{Contributions.}
Having established these natural and well-studied notions of better-than-worst-case optimality, we now present our results. 
We first prove an adaptive adversarial lower bound showing that no deterministic $o(\log \log n)$-competitive algorithm exists. 
We then match this bound up to constant factors by providing an $\Theta(\log \log n)$-competitive algorithm. 
Similarly, we show an adversarial lower bound proving that no deterministic $o(\log n)$-order-competitive algorithm exists, and we match this bound up to constants by giving an $\Theta(\log n)$-order-competitive algorithm. 

We further extend the approach in two directions. 
First, we consider a variant of universal optimality that was originally proposed by Cardinal et al.~\cite[STOC'10]{cardinal_sorting_2010} who consider the following question: ``If one is allowed to preprocess the universally fixed graph $G$, can one then design competitive algorithms?''
We show that, for any graph $G \in \bG_n$, one can preprocess $G$ in $O(m)$ time ($m$ is the number of cliques in $G$) to obtain an $O(1)$-competitive algorithm. 
Second, we study the related \emph{set intersection} problem. 
Here we obtain a strict separation: for set intersection, there is no $o(\log n)$-competitive algorithm, regardless of whether the fixed quantity is the undirected duplicate graph or the induced DAG. 
We conclude with a simple $O(\log n)$-competitive algorithm for these settings, yielding a detailed classification of the difficulty of element distinctness and set intersection under various inputs. 

We thereby present a comprehensive and extensive study of the complexity of the element distinctness and set intersection problems for when the input contains duplicates. 
Our bounds cover a variety of popular models for better-than-worst-case analysis and are tight.

\section{Preliminaries}

We study the \emph{element distinctness} and \emph{set intersection} problems in a comparison model of computation.
For element distinctness, the input is a list $I$ of $n$ elements drawn from some totally ordered set $U$. 
Comparing two elements $x$ and $y$ returns one of ``$x < y$'', ``$x = y$'', or ``$x > y$'', and this is the only operation permitted on the elements of $I$. 
The goal is to output any pair of equal elements, or report that no such pair exists, using as few comparisons as possible.  

In the set intersection problem, the input consists of two lists $A$ and $B$, each containing $n$ elements from the ordered universe $U$. 
The task is to output a pair of elements $a \in A$ and $b \in B$ such that $a = b$, or report that no such pair exists.  

\paragraph{Algorithms and runtime.}
Our lower bounds uniquely count the number of comparisons required before the algorithm can find a witness.  
Naturally, the number of comparisons performed provides a lower bound on the total running time. Our upper bound algorithms require additional operations such as reading the input. However, if the input $I$ is provided as a stream, the running time our algorithms is dominated by the number of comparisons they perform.  
Accordingly, all our bounds can be expressed in terms of algorithmic running time.  
We assume that algorithms are \emph{deterministic}.  
Because the output size is a constant witness, randomized algorithms may behave very differently.
Consider the example where $I$ consists of $\frac{n}{2}$ copies of the same value and $\frac{n}{2}$ distinct values.  
An adversarial argument gives a deterministic lower bound of $\Omega(n)$ comparisons, as the algorithm can be forced to inspect all distinct elements first.
However, an algorithm that randomly selects $(i, j)$ and compares $I[i]$ and $I[j]$ terminates in expected constant time.

\paragraph{Undirected graphs and quantities.}
An instance $I$ of element distinctness induces an undirected graph $G(I)$ on $n$ vertices.  
The vertex set is $[n] = \{1, \dots, n\}$, and for $x, y \in [n]$, there is an edge $xy$ if $I[x] = I[y]$.  
Similarly, an instance $(A, B)$ of set intersection induces an undirected graph $G(A, B)$ on $2 n$ vertices.  
All connected components of $G(I)$ and $G(A,B)$ are cliques. 
We refer to these components the \emph{clusters} of $I$, or of $(A, B)$.  
The \emph{size} of a cluster is its number of vertices.
We denote by $\bG_n$ the set of all such $n$-vertex graphs. 
For a fixed $G \in \bG_n$ and integer $L$, we define $C(L)$ as the \emph{total size} of all clusters of size smaller than $L$, and $D(L)$ as the \emph{number} of clusters of size at least $L$.

\newpage
\section{Algorithms for element distinctness}
We present two comparison-based procedures for element distinctness, which we run in parallel, terminating as soon as either procedure completes.
The first, \texttt{Block Sorting}, operates on fixed-size batches of elements.  Given a
parameter $k$, it repeatedly sorts a block of $k$ elements and halts if a
duplicate appears within the block; otherwise it discards the block and
continues (Algorithm~\ref{algo:block_sorting}).
The second, \texttt{Median Recursion}, uses a parameter $L$.  It
recursively partitions the input at the median and prunes any recursive
branch whose subproblem has fewer than $L$ elements
(Algorithm~\ref{algo:median_recursion}).

Intuitively, \texttt{Median Recursion} is preferable when the universal
running time lies between $O(n)$ and $O(n\log n)$, whereas
\texttt{Block Sorting} can terminate sooner on instances where the
universal running time is sublinear.
To analyze the universal running time of these algorithms, we fix a graph $G$ and an arbitrary value $L \geq 2$. The graph $G$ induces two quantities:
$C(L)$ is the \emph{total size} of clusters of size less than $L$,
and $D(L)$ is the \emph{number} of clusters of size at least $L$.
Let $I$ be an input with $G(I)\cong G$.
We now fix an arbitrary $L$ and run \texttt{Median Recursion} on $I$ with this
parameter.  In addition, suppose $D(L)$ is known.  We run
\texttt{Block Sorting} on the same input with $k = 2D(L)$.  With this
choice, the running times of both algorithms can be expressed in terms
of $n$, $C(L)$, and $D(L)$.
If $G$ were known in advance, we could choose $L$ to minimize the
overall running time.  We show that, even without knowledge of $G$, it
suffices to try $O(\log\log n)$ values of $k$ and $L$ in parallel to
obtain a universally $O(\log\log n)$-competitive algorithm.

\textbf{Analysis of Block Sorting.}
We analyze the universal running time of Algorithm~\ref{algo:block_sorting}.
Fix $G\in\bG_n$ and $L \ge 2$ with $C(L) < n/2$. Then the values $C(L)$ and $D(L)$ are fixed and we assume that $k=2D(L)$.

\begin{lemma}\label{lemm:block_sorting}
\texttt{Block Sorting} finds a duplicate after $O\big((C(L)+D(L))\max\{1, \log D(L)\}\big)$ comparisons.
\end{lemma}

\begin{proof}
The algorithm can terminate only (i) by returning a duplicate, or (ii) by giving up.
We first show that, with $k=2D(L)$ and $C(L)<n/2$, the algorithm, if it does not give up, will return a duplicate after the appropriate number of comparisons.
We then show that it will not give up. 

Assume for contradiction that the algorithm performs $1+\lceil C(L)/D(L)\rceil$
iterations without terminating. In iteration $i$, let
$S_i$ be the set of elements inspected by the algorithm, then $|S_i|=k=2D(L)$. Since no duplicate is found within $S_i$, all its elements must belong to distinct clusters. Since there
are only $D(L)$ clusters of size $\ge L$, at least $D(L)$ elements of
$S_i$ must lie in clusters of size $<L$. After $1+\lceil C(L)/D(L)\rceil$ iterations, this accounts for $>C(L)$ distinct input elements that all lie in
clusters of size $<L$ -- a contradiction. Hence, a duplicate is found within at
most $1+\lceil C(L)/D(L)\rceil$ iterations if the algorithm does not give up. 
Each iteration sorts $k$ elements, requiring $O(k\log k)=O(D(L)\max\{1, \log D(L)\})$
comparisons. With at most $1+\lceil C(L)/D(L)\rceil$ iterations, the total number of comparisons is
$O\big((C(L)+D(L))\max\{1, \log D(L)\}\big)$.  
\begin{figure}[H]
  \centering
  \caption{Our algorithms \texttt{Block Sorting} and \texttt{Median Recursion}  for element distinctness. 
  }
  \begin{subfigure}{0.48\linewidth}
    \caption{\texttt{Block Sorting}} \label{algo:block_sorting}
    \centering
    \begin{algorithmic}[1]
       \Require Input $I$, integer $k \ge 1$
    \While{$|I| \ge 2k$}
      \State $S \gets k \text{ arbitrary elements from } I$
      \State Sort $S$
      \State If $S$ contains a duplicate, \textbf{return} it
      \State $I \gets I - S$
    \EndWhile
    \State Sort $I$
    \State If $I$ contains a duplicate, return it
    \State Give up and \textbf{terminate}.
    \end{algorithmic}
  \end{subfigure}\hfill
  \begin{subfigure}{0.48\linewidth}
    \caption{\texttt{Median Recursion}}\label{algo:median_recursion}
  \captionsetup{position=top}
    \centering
    \begin{algorithmic}[1]
           \Require Input $I$, integer $L \ge 1$
    \If {$|I| \ge L$}
      \State $m \gets \operatorname{median}(I)$
      \State Check if $m$ occurs multiple times in $I$
      \State If so, \textbf{return} $m$
      \State $\operatorname{median-recursion}(\{x \in I \ |\ x < m\}, L)$
      \State $\operatorname{median-recursion}(\{x \in I \ |\ x > m\}, L)$
    \EndIf
    \State Give up and \textbf{terminate}
    \end{algorithmic}
  \end{subfigure}
  \vspace{-2em}
\end{figure}

\newpage
Finally, suppose for the sake of contradiction that $C(L)<n/2$ but the algorithm still gives up (i.e., the algorithm processes all
elements without ever finding a duplicate inside a block of size $k$). In every iteration
and in the final sorting step,  at least $D(L)$ inspected elements must come
from clusters of size $<L$, Hence, at least  half of all inspected
elements lie in clusters of size $<L$, implying $C(L)\ge n/2$ -- a contradiction.
\end{proof}

\textbf{Analysis of Median Recursion.}
We analyze the universal running time of Algorithm~\ref{algo:median_recursion}.
Fix $G\in\bG$ and some $L \ge 2$. Let $C(L)$ be the total size of clusters of size $<L$ and
let $D(L)$ be the number of clusters of size $\ge L$. Suppose that $C(L) < n$.

\begin{lemma}\label{lemm:median_recursion}
\texttt{Median Recursion} finds a duplicate after  $O\big(n + C(L)\max\{1,\log(C(L)/L)\}\big)$ comparisons.
\end{lemma}

\begin{proof}
Algorithm~\ref{algo:median_recursion} partitions an input $I'$ by comparing
against the median. Either the algorithm finds an element in $I'$ that equals the median (and the algorithm
terminates), or every cluster in $I'$ remains intact after the partition. 
We refer to a recursive call on $I'$ with $|I'|<L$ as a \emph{small call}.
In a small call, all elements of $I'$ belong to
clusters of size $<L$. Distinct small calls are disjoint, and each small call cannot contain any cluster of size $\geq L$, so the total input size across all small calls is at most $C(L)$.
Moreover, since we partition at the median, there can be at most $O(C(L) / L)$ small calls. 

Suppose that $k$ small calls have terminated before our algorithm terminates, and consider the recursion tree. 
The recursion is depth-first and stops after each small call. 
Hence, the recursion tree has $k$ leafs which form a balanced binary tree $T$ whose height is $O(\log k)$, followed by a path to the root call.
Each call to the \texttt{Median Recursion} algorithm performs $O(|I'|)$ comparisons on its subproblem $I'$. 
Thus, the total number of comparisons done on the path to the root call is $O(n)$. 
The work of calls in $T$ is bounded by the total input size of each level, multiplied
by the number of levels. 
Since all elements in $T$ belong to clusters of size less than $L$, the total input size of each level is at most $C(L)$. 
Since $k \in O(C(L) / L)$, it follows that there are at most 
$O\big(C(L)\max\{1,\log(C(L)/L)\}\big)$ comparisons done across all calls in $T$, which implies the claimed bound. 
Since $C(L)<n$, at least one cluster has size $\ge L$, so the
algorithm must eventually expose a duplicate and terminate.
\end{proof}

\subsection{Clairvoyant and Oblivious Algorithms}
If the universal graph $G$ were known, then for any $L$, we could compute $C(L)$ and $D(L)$ and select $L$ to minimize the running time. This yields a \emph{clairvoyant} algorithm that we claim is universally optimal for a fixed $G$:

\begin{theorem} \label{theo:clair_median}
  For every $G$, there is an algorithm $\cA_G$ for element distinctness with
  \[
    \universal(\cA, G) \in O\Big(\min\Big(\displaystyle \min_{\substack{L\ge 2\\ C(L) < n}} (n + C(L) \log \frac{C(L)}{L}), \min_{\substack{L \ge 1\\ C(L) < \frac{n}{2}}} (C(L) + D(L)) \max\{1, \log D(L)\}\Big)\Big).
  \]
\end{theorem}

Next, we design an algorithm that is oblivious to $G$.
It runs multiple instances of \texttt{Median Recursion} and \texttt{Block Sorting} on the same input $I$, each with a different parameter choice, and stops when any instance terminates with a witness for \texttt{yes}. We prove that, for any input $I$ with induced graph $I(G)$, there exists an instance whose running time asymptotically matches the running time of the clairvoyant algorithm for the graph $G \cong I(G)$.
We run $O(\log\log n)$ instances in parallel. Thus a sequential simulation of this parallel algorithm incurs at most an $O(\log\log n)$ multiplicative overhead over the clairvoyant algorithm.

\begin{theorem} \label{theo:algo_oblivious}
  There is an algorithm $\cA$ for element distinctness with
  \[
    \hspace{-0.5cm}\universal(\cA, G) \in O\Big(\log\log(n) \cdot \min\Big(\displaystyle \min_{\substack{L \ge 2\\ C(L) < n}} \big(n + C(L) \log \frac{C(L)}{L}\big), \min_{\substack{L \ge 1\\ C(L) < \frac{n}{2}}} (C(L) + D(L)) \max\{1, \log D(L)\}\Big)\Big).
  \]
\end{theorem}

We prove the theorem via two claims, on \texttt{Block Sorting} and \texttt{Median Recursion}, respectively. 
 
\begin{claim} \label{clai:block_sorting_oblivious}
  There is an algorithm $\cA$ for element distinctness with
  \[
    \universal(\cA, G) \in O\Big(\log\log(n) \cdot \min_{\substack{L \ge 1\\ C(L) < \frac{n}{2}}} (C(L) + D(L)) \max\{1, \log D(L)\}\Big).
  \]
\end{claim}
\begin{proof}
  Let $I$ be the input. The algorithm $\cA$ runs the following $O(\log \log n)$ algorithms in parallel:
  \begin{itemize}
    \item for $i=0,\dots,\lceil\log\log n\rceil$, a \texttt{Block Sorting} instance with $k=2\cdot 2^{2^{i}}$; and
    \item  a doubling scheme that starts with $\ell=2$, checks $\min\{n,\ell\}$ arbitrary elements $I' \subseteq I$ for a duplicate via sorting $I'$ (this checking step may pick elements that were checked before), and then doubles $\ell$.
  \end{itemize}

Let $L_0=\arg\min_{L \ge 1}^{C(L)<n/2} \big(C(L)+D(L)\big)\max\{1, \log D(L)\}$. We show that some parallel branch halts within $O\big((C(L_0)+D(L_0))\max\{1, \log D(L_0)\}\big)$ comparisons.

If $C(L_0)\ge D(L_0)^2$, pick $i=\lceil\log\log D(L_0)\rceil$.
Then $k=2\cdot 2^{2^{i}}$ satisfies $2D(L_0)\le k\le 2D(L_0)^2\le 2C(L_0)$,
and Lemma~\ref{lemm:block_sorting} yields a duplicate in $O\big((C(L_0)+k)\log k\big)\subseteq O\big(C(L_0)\max\{1, \log D(L_0)\}\big)$ comparisons.
Otherwise $C(L_0)\le D(L_0)^2$. When the doubling scheme reaches $\ell=2^{\lceil\log (D(L_0)+1)\rceil}$, then by the pigeonhole principle, any subset $I' \subset I$ of size $\ell$ must contain a duplicate. This detects a duplicate in $O(\ell\log \ell) \subseteq O\big(D(L_0)\max\{1, \log D(L_0)\}\big)$ comparisons.
\end{proof}

\begin{claim}\label{clai:median_recursion_oblivious}
There exists an algorithm $\cA$ with
\[
  \universal(\cA,G)\in
  O\Big(\log\log n\cdot
    \min_{\substack{L\ge 2\\ C(L)<n}}
      \big(n+C(L)\log\tfrac{C(L)}{L}\big)
  \Big).
\]
\end{claim}

\begin{proof}
  Let $L_0 = \arg \min_{L \ge 2}^{C(L) < n} (n + C(L) \log \frac{C(L)}{L})$.
  If we guess a $C \ge C(L_0)$ and an $L \le L_0$,
  then running median recursion with this $L$ until there are $C$ elements among recursive calls with $|I'| \le L$
  will successfully find a duplicate.
  If we have $C \le 2 C(L_0)$ and $L \ge \frac{L_0^2}{C(L_0)}$,
  then $O(n + C \log \frac{C}{L}) = O(n + C(L_0) \log \frac{C(L_0)}{L_0})$.
  To guess such a $C, L$, algorithm $\cA$ will try $O(\log \log n)$ different values of $\log \frac{C(L)}{L}$ in parallel
  and use a doubling strategy to guess $C$, see Algorithm~\ref{algo:oblivious_median}.

  We now analyze Algorithm~\ref{algo:oblivious_median}.
  Let $C = 2^{\lceil \log(C(L_0)) \rceil}$ and $i = 2^{\lceil\log\log \frac{2 C(L_0)}{L_0}\rceil}$.
  Put $L = \max(2, C / 2^i)$.
  Then, $C \in [C(L_0), 2 C(L_0)]$ and $L \le L_0$,
  but also $\log \frac{C}{L} \in O(\log \frac{C(L_0)}{L_0})$.
  Therefore, the $(i, C)$-iteration of Algorithm~\ref{algo:oblivious_median} will find a duplicate,
  in $O(C \log \frac{C}{L}) = O(C(L_0) \log \frac{C(L_0)}{L_0})$ comparisons.
  (We do not need $O(n)$ time since we remember the top part of the recursion tree.)
  The number of comparisons of previous $(i, *)$ iterations forms a geometric series,
  which sums to $O(n + C \log \frac{C}{L}) = O(n + C(L_0) \log \frac{C(L_0)}{L_0})$.
  Since Algorithm~\ref{algo:oblivious_median} tries $O(\log \log n)$ different values of $i$ in parallel, this shows the claim.
\end{proof}
\cref{theo:algo_oblivious} now follows by sequentially simulating the parallel execution of \Cref{clai:block_sorting_oblivious} and \ref{clai:median_recursion_oblivious}. 

\begin{algorithm}[H]
  \begin{algorithmic}
    \Require Input $I$
    \For {$i = 1, 2, 4, 8, \dots, \log(n)$} in parallel
      \For {$C = 1, 2, 4, 8, \dots$}
        \State Run median recursion with $L \gets \max(2, C / 2^i)$,
        \State until there are $\ge C$ elements among recursive calls with $|I| < L$.
        \State Remember the top $\lfloor \log \frac{n}{C}\rfloor$ levels of the recursion tree and reuse them for the next iteration.
      \EndFor
    \EndFor
  \end{algorithmic}
  \caption{Oblivious Median Recursion} \label{algo:oblivious_median}
\end{algorithm}

\newpage
\section{Lower Bounds for Element distinctness} \label{sect:no_competitive}

We first show that, with respect to universal optimality,
no algorithm can be $o(\log \log n)$-competitive.

\subsection{No element distinctness algorithm can be \texorpdfstring{$\pmb{o(\log \log n)}$}{o(log log n)}-competitive}
\label{sec:no_loglogn}

Let $\cA$ be an arbitrary algorithm for element distinctness. 
We show a lower bound that uses an adaptive adversary to prove that $\cA$ cannot be universally $o(\log \log n)$-competitive for every graph $G$. 
Formally, let $n \in \mathbb{N}$ be arbitrary.
  We consider the following game between the algorithm and an (adaptive) adversary on a list $X$ of $n$ labeled elements.
  In each round, the algorithm compares two elements,
  and the adversary has to answer with ``<'', ``='' or ``>''.
  The goal of the adversary is to survive for $\frac{1}{8} n \log \log n$ rounds without ever answering ``=''.
  Afterwards, the adversary will construct the input $I$ from $X$ by picking an ordered set $U$
  and assigning each element of $X$ a value in $U$, so that the pairwise comparisons on $I$
  are consistent with the answers it has given. It then sets $G = G(I)$.
  This ensures that $\universal(\cA, G_j) \ge n \log \log n$.
  Finally, the adversary has to pick an algorithm $\cA'$ with $\universal(\cA', G) \in O(n)$. 
  
  \textbf{The strategy.}
  The adversary uses the following strategic game to answer the comparisons:
  Imagine a complete binary tree $T$ of height $n \log n$, where each elements $x \in X$ is stored at some node $u(x)$ in the tree.
  Initially, all elements are placed at the root of $T$. 
  In the end, the adversary will pick $U$ to be the set of leaves of $T$, ordered from left to right.
  When the algorithm compares two elements $x, y \in X$, the adversary answers as follows:
  \begin{itemize}[noitemsep] \setlength{\itemsep}{3pt}
    \item If there exists a node $u \in T$ such that $u(x)$ lies in the left subtree of $u$ and $u(y)$ lies in the right subtree of $u$, answer ``$x < y$''.
    \item If there exists a node $u \in T$ such that $u(x)$ lies in the right subtree of $u$ and $u(y)$ lies in the left subtree of $u$, answer ``$y < x$''.
    \item If $u(x) = u(y) = u$, move $u(x)$ to the left child of $u$ and $u(y)$ to the right child. Answer ``$x < y$''.
\item If $u(x)$ and $u(y)$ lie on the same root-to-leaf path and $u(y)$ is in the right subtree of $u(x)$, put $u = u(x)$, move $u(x)$ to the left child of $u$, and answer ``$x < y$''. If $u(y)$ is in the left subtree of $u(x)$ instead, put $u = u(x)$, move $u(x)$ to the right child of $u$, and answer ``$y > x$''.
\item The case where $u(x), u(y)$ lie on the same root-to-leaf path and $u(y)$ is above $u(x)$ is symmetric.
  \end{itemize}
  Note that in each case, the adversary moves at most two elements down the tree and these elements move one step. 
  After $\frac{1}{8} n \log \log n$ rounds of this game, not too many elements can be deep in the tree:

 \begin{claim} \label{clai:few_deep}
    After $\frac{1}{8} n \log \log n$ rounds have passed,
    there exists an integer $i \in \{\lfloor\frac{1}{2} \log \log n\rfloor, \dots, \lfloor \log \log n \rfloor\}$
    such that there are strictly fewer than $n / 2^i$ elements of depth at least $2^i$.
  \end{claim}

  \begin{proof}
  We show the claim via an indirect proof: If there is no such $i$,
  then the total depth of all elements is at least
  \[
    \sum_{i=\lfloor\frac{1}{2} \log \log n\rfloor}^{\lfloor\log \log n\rfloor} 2^i \cdot (\frac{n}{2^{i}} - \frac{n}{2^{i+1}}) = \sum_{i=\lfloor \frac{1}{2} \log \log n\rfloor}^{\lfloor\log \log n\rfloor} \frac{n}{2} = \frac{n}{2} \Big(\lfloor\log \log n\rfloor - \big\lfloor\frac{1}{2} \log \log n\big\rfloor +1\Big) > \frac{n}{4} \log \log n,
  \]
  so more than $\frac{1}{8} n \log \log n$ rounds have passed. This shows the claim.
\end{proof}

The core idea is as follows: the adversary plays the game without ever revealing a duplicate. After $\tfrac{1}{8} n \log \log n$ rounds, it collects all comparisons made so far.
It then chooses an integer $L$ and defines a graph $G$ on $X$ that contains many clusters of size $L$, formed by grouping elements that share a root-to-leaf path.
It then moves all elements in a cluster down to the leaf of this path.
Let $U$ be the set of leaves of $T$, ordered from left to right.
If $X = (x_1, \dots, x_n)$, put $I = (u(x_1), \dots, u(x_n))$, then $G(I) \cong G$.
Since elements are only moved downward, the pairwise comparisons on $I$ remain consistent with the adversary’s answers.
Hence, we obtain an input $I$ of $n$ elements from $U$, on which, after $\tfrac{1}{8} n \log \log n$ comparisons, algorithm $\cA$ has not yet found a duplicate, although one exists.
Finally, we show that a clairvoyant algorithm, knowing $G$, can solve all inputs $I'$ with $G(I') \cong G$ in linear time.
Therefore, for this graph $G$, there exists an algorithm with linear universal running time, whereas $\cA$ has universal running time $\Omega(n \log \log n)$.

\begin{theorem} \label{theo:lower_competitive}
    For every algorithm $\cA$ and every positive integer $n$,
    there exists an $n$-vertex graph $G$ and an algorithm $\cA'$ such that
    $\universal(\cA, G) \ge \tfrac{1}{8} n \log\log n \text{, but}
    \universal(\cA', G) \in O(n).$
\end{theorem}
\begin{proof}
    Let $X$ be a list of $n$ labeled elements. 
    Run $\cA$ on $X$ and play the adversarial game that we described in the beginning of the section. 
    After $\frac{1}{8} n \log \log n$ rounds, let $i$ be as in Claim~\ref{clai:few_deep} and set $L = \frac{n}{2^{2^{i-1}}} > 1$.
    
    To defined the input $I$, let $U$ the set of leaves of $T$, ordered from left to right.
    To turn $X$ into $I$, we move each element of $X$ into a leaf of $T$.
    More precisely:
    Observe that, after $\frac{1}{8} n \log \log n$ comparisons, there can be no element in $X$ where $u(x)$ is in a leaf of $T$. 
    While there is an unoccupied leaf in $T$ with root-to-leaf path $\pi$ for which at least $L$ elements $x \in X$ have $u(x)$ in $\pi$, pick an arbitrary subset $X'$ of $L$ of these elements. 
    Intuitively, we form $X'$ into a cluster of size $L$. 
    Formally, move all elements in $X'$ down along $\pi$ to the unoccupied leaf. 
    Repeat this procedure until no such path exists anymore. 
    Then, one-by-one, move all remaining elements down to an unoccupied leaf of $T$.
    Put $I = (u(x_1), \dots, u(x_n))$ where $X = (x_1, \dots, x_n)$ then $I$ is a list of elements from $U$ and all answers given by the strategic game are consistent with $U$. 
    Since the  game never answers equality, the algorithm $\cA$ did not find a duplicate on $I$ after $\frac{1}{8} n \log \log n$ comparisons.  

    Let $G := G(I)$. 
    Since all elements in $G$ either lie in a cluster of size $L$, or in a singleton cluster, it follows that $C(L)$ equals the number of singleton elements. 
    We use the original state of the tree $T$ to bound this quantity.
    By \cref{clai:few_deep}, our tree $T$ had fewer than $\frac{n}{2^{i-2}}$ elements of depth at least $2^{i-2}$.
    On the other hand, the elements of depth $\le 2^{i-2}$ lie at most $\le 2^{2^{i-2}}$ distinct paths from the root. 
    This counting argument shows that there are at most $L \cdot 2^{2^{i-2}}$ such elements that were not packed into a cluster, hence 
    \[
    C(L) \le \frac{n}{2^{i-2}} + L \cdot 2^{2^{i-2}} \le \frac{n}{2^{i-2}} + \frac{n 2^{2^{i-2}}}{2^{2^{i-1}}} = \frac{n}{2^{i-2}} + \frac{n}{2^{2^{i-2}}} \le \frac{n}{2^{i-3}},
    \]
    First, this implies $C(L) < n$ since $i \ge \lfloor \frac{1}{2} \log \log n \rfloor$.
    Second, this implies
    \[
    C(L) \log \frac{C(L)}{L} \le C(L) \log \frac{n}{L} \le \frac{n}{2^{i-3}} \log 2^{2^{i-1}} = 4 n.
    \]
    We choose $\cA'$ to be the \texttt{Median Recursion} algorithm with parameter $L$. 
    By \cref{lemm:median_recursion}, for any input $I'$ such that $G(I') \cong G$ the running time of this algorithm is $O(n)$.
\end{proof}

\subsection{A \texorpdfstring{$\pmb{\log \log n}$}{log log n}-competitive algorithm}

We complement \cref{theo:lower_competitive}
by showing that the algorithm in \cref{theo:algo_oblivious} is $O(\log \log n)$-competitive.
In particular, we prove a lower bound using an adaptive adversary, showing that for all graphs $G$, no algorithm $\cA$ can achieve a better running time than the corresponding clairvoyant algorithm from \cref{theo:clair_median}.

Fix some $G \in \bG$ that contains at least one cluster of size at least two. Then, for every choice of $L$, this induces some $C(L)$ and $D(L)$. 
Consider any algorithm $\cA$.
The adversary employs the strategic game from Section~\ref{sec:no_loglogn} on a list $X$ of $n$ labeled elements.
We compute some $M$ depending on the values of $C(L)$ and $D(L)$ for every $L$. 
The adversary then plays the game for $M$ rounds without revealing a duplicate. 
After $M$ rounds, all elements of $X$ are in some node of a tree $T$. The adversary then moves all elements from $X$ to the leaves of $T$ in a specific manner. The input $I[i]$ is the leaf that contains the $i$'th element from $X$ and we prove that $G(I) \cong G$.
The pairwise comparisons on $I$ are automatically consistent the adversary's answers. 
Thus, for this input $I$, the algorithm $\cA$ fails to find a duplicate after $M$ comparisons,
even though $I$ contains at least one pair of equal elements.

\subsubsection{A lower bound complementing \cref{lemm:median_recursion}}
We first show a lower bound that complements \cref{lemm:median_recursion},
but that is missing the $O(n)$-term.
This already shows that our algorithm from Theorem~\ref{theo:algo_oblivious} is $O(log \log n)$-competitive for inputs that require at least a linear number of comparisons.

  \begin{claim}
    \label{clai:isomorphic}
    Fix some $G \in \bG$ and define $M := \min_{L\ge 2}^{C(L) < n} \left( \frac{1}{4} C(L) \cdot \max \{0, \log_2 \frac{C(L)}{2L} \} \right)$.
    Play the adversarial game on a list $X$ of size $n$.
    Let at most $M$ rounds have passed and let be $T$ the current state of the tree.
    Then, there exists a graph $G'$ on $X$ that is isomorphic to $G$
    such that for each cluster in $G'$, all elements share a root-to-leaf path in $T$. 
  \end{claim}

\begin{proof}
    Recall that any graph $G \in \bG$ consists of only edges that are in a cluster (clique). 
      We consider the clusters in $G$ in decreasing order of size. Let the current cluster that we consider have size $L$. 
      We find some root-to-leaf path that contains at least $L$ unused elements, use these elements to form a cluster of size $L$, and remove them from $T$. 
      This procedure, if successful, creates a graph $G'$ on $X$ that is isomorphic to $G$. 

    What remains is to prove that, for all $L$ that we consider, there exists such a root-to-leaf path. 
    If there exists an integer $L \geq 2$ with $C(L) < n$, such that $\log_2 \frac{C(L)}{2L} \le 0$, then $M \leq 0$. 
    However, this means we have played zero rounds, so all elements of $X$ are at the root of $T$ and we can create any graph on $X$, which shows the claim. 
    Therefore, suppose that for all $L \geq 2$, we have $\log \frac{C(L)}{2L} > 0$.
    When considering a cluster of size $L$, there are per definition at least $C(L)$ unused elements in the tree.
    Also, since at most $\frac{1}{4} C(L)\log \frac{n}{2L}$ rounds have passed,
    there are at most $\frac{1}{2} C(L)$ elements of depth $\ge \log \frac{n}{2L}$.
    Hence, there are at least $\frac{1}{2} C(L)$ unused elements of depth $\le \lfloor \log \frac{C(L)}{2L} \rfloor$.
    These elements lie on $2^{\lfloor \log \frac{C(L)}{2L} \rfloor} \le \frac{C(L)}{2L}$
    different root-to-leaf paths.
    Thus, one of these paths contains at least $\frac{C(L)}{2} / \frac{C(L)}{2L} =  L$ elements.
\end{proof}

\noindent
Claim~\ref{clai:isomorphic} immediately implies the following lower bound. 

\begin{proposition} \label{theo:lower_median_recursion}
  Let $G \in \bG$ and define $M := \min_{L\ge 2}^{C(L) < n} \left( \frac{1}{4} C(L) \cdot \max \{0, \log_2 \frac{C(L)}{2L} \} \right)$.
  Then every algorithm $\cA$ satisfies: 
  \[
  \universal(\cA, G) \ge M.
  \]
\end{proposition}

\begin{proof}
    For any algorithm $\cA$, the adversary plays $M$ rounds
    and defines $U$ to be the set of leaves of $T$.
    Then, it applies Claim~\ref{clai:isomorphic} and pushes all elements in a cluster down to the same distinct leaf.
    This defines input $I$ with $G(I) \cong G$, and all answers given by the adversary are consistent with $I$.
    It follows that, on this input, $\cA$ did not find a duplicate after $M$ comparisons,
    even though there is a duplicate.     
\end{proof}

\subsubsection{A lower bound complementing \Cref{lemm:block_sorting}}
We use our adversarial game to construct an additional lower bound that complements the running time of \Cref{lemm:block_sorting}. 

\begin{definition} \label{def:G_prime}
    Let $G \in \bG$ be an $n$-vertex graph with at least two clusters.
    We define by $G'$ the graph that is obtained by iteratively deleting the largest cluster from $G$ until $G'$ has $n'$ vertices with $n' \le \frac{3}{4} n$. Define for all $L$, the value $C'(L)$ as the total size of clusters in $G'$ of size strictly less than $L$, and by $D'(L)$ the number of clusters of size at least $L$.  
\end{definition}

Consider the adversarial game on a list $X = (x_1, \dots, x_n)$ of $n$ elements.
If this game lasts $\Omega(n)$ rounds, then \cref{theo:lower_median_recursion} is tight.
Therefore, suppose that this game lasts at most $\frac{n}{8}$ rounds. 
We preset a different construction for transforming $X$ into an input $I$ with $I(G) \cong G$,
that yields a stronger bound in this regime.
The idea behind this construction is as follows: denote by $T$ the state of our binary tree after at most $\frac{n}{8}$ rounds. 
Recall that, in each round, the adversary selects at most two elements $x \in X$ and moves $u(x)$ down the tree by one step.
It follows that there are still  $\frac{3}{4} n$ elements in $X$ for which $u(x)$ is at the root.
Since $n' \le \frac{3}{4} n$, we can always realize $G'$ on these elements. However, we can do even better
and use $G'$ to get rid of all elements that are not stored at the root.

\paragraph{Our reconstruction algorithm. }
  First, we consider the graph $G'$ from Definition~\ref{def:G_prime} and we iterate over its clusters decreasing order by size.
  For each cluster $C$, we pick a non-root node $v \in T$ of minimal depth that contains at least one element. 
  If there are at least $|C|$ elements $x \in X$ for which $u(x) = v$, then we use $|C|$ of these elements to form a cluster of size $|C|$, and remove them from $T$.
  If there are fewer than $|C|$ such elements, then we supplement this set by adding elements $x \in X$ where $u(x)$ is the root of $T$.
  Since $n' \le \frac{3}{4} n$, this is always possible.
  Finally, if we run out of non-root nodes of $T$, then 
  all remaining clusters of $G$ (including those not in $G'$) are formed by iteratively selecting an arbitrary
  set of elements $x \in X$ with $u(x) = \Root(T)$.
  
  \begin{claim} \label{clai:xmas_magic}
    If our reconstruction algorithm can form all clusters in $G'$ without running out of non-root elements,
    then the total depth of all elements is $>\frac{1}{16} \min_{L \ge 1}^{C'(L) < n'} ((C'(L) + D'(L)) \max\{1, \log D'(L)\})$
  \end{claim}

  \begin{proof}
    Let $T$ be the original state of our binary tree. 
    We call an element $x \in X$ a \emph{root} element if $u(x)$ is in the root of the original tree $T$ and a non-root element otherwise. 
       We call a cluster in $G'$ \emph{whole} if its corresponding cluster in $X$ consists exclusively of non-root elements. 
       We call the cluster \emph{split} otherwise. 
  We make the following observations:
  \begin{enumerate}[nolistsep]
    \item If there have been $k$ split clusters so far, then all non-root elements have depth $\ge \max\{1, \log k-1\}$.
    \item Every whole cluster consists of only non-root elements.
    \item Every split cluster contains at least one non-root element.
  \end{enumerate}
  Let $U_1, \dots, U_m$ be the clusters in $G'$ in increasing order by size.
  Pick the smallest integer $i$ such that at least half the clusters in $U_i, U_{i+1}, \dots, U_m$ are split.
  Put $L = |U_i|$, then there are at least $\frac{1}{2} D'(L)$ such split clusters.
  Thus, by (1) and (3), the total depth of all elements is at least $\frac{1}{4} D'(L) \max\{1, \log (\frac{1}{4} D'(L)) - 1\}$.
  On the other hand, since $i$ is minimal,
  for any $j \le i$, at most half the clusters $U_j, U_{j+1}, \dots, U_{i-1}$ are split.
  Therefore, among $U_1, \dots, U_{i-1}$, we can match each split cluster to a distinct non-smaller whole cluster.
  Hence, the total size of whole clusters among $U_1, \dots, U_{i-1}$ is at least $\frac{1}{2} C'(L)$.
  Thus, by (1) and (2), the total depth of all elements is at least $\frac{1}{2}C'(L) \max\{1, \log(\frac{1}{2} D'(L)) - 1\}$.
  Putting things together, the total depth of all elements is 
  \[
    \ge (\frac{1}{4} D'(L) + \frac{1}{2} C'(L)) \max\{1, \log (\frac{1}{4} D'(L)) - 1\} \ge \frac{1}{16} (D'(L) + C'(L)) \max\{1, \log D'(L)\},
  \]
  Since $L = |U_i|$, we have $L \ge 1$ and $C'(L) < n'$,
  \end{proof}

\begin{lemma} \label{lemm:linear_subset}
  Let $G \in \bG$ be an $n$-vertex graph with at least two clusters and $G'$ be the graph from Definition~\ref{def:G_prime}.
  Then:

  \[
  \hspace{-1.0cm} \min(n'/8, \frac{1}{32} \min_{\substack{L \ge 1\\C'(L) < n'}} (C'(L) + D'(L)) \max\{1, \log D'(L)\}) \ge \frac{1}{1000} \min(n, \min_{\substack{L \ge 1\\C(L) < n/2}} (C(L) + D(L)) \max\{1, \log D(L)\}).
  \]
\end{lemma}
\begin{proof}
  We note that $G'$ is non-empty since $G$ consists of at least two clusters.
  Let $L_1$ be the size of the smallest cluster that was deleted from $G$ (i.e., the last cluster that we deleted before we obtained $G'$) and let
  \[
  L_0 = \arg \min_{\substack{L \ge 1\\C'(L) < n'}} (C'(L) + D'(L)) \max\{1, \log D'(L)\}).
  \]
  Then $L_0 \le L_1$ since $C'(L_0) < n'$, so $C'(L_0) = C(L_0)$.
  Suppose first that $L_1 \ge \frac{n}{4}$. Then, we deleted a single cluster of size $L_1$,
  so $D'(L_0) = D(L_0) - 1$. This implies
  \[
      (C'(L_0) + D'(L_0)) \max\{1, \log D'(L_0)\}) \ge \frac{1}{4}(C(L_0) + D(L_0)) \max\{1, \log D(L_0)\} .
  \]
  Otherwise, we have $L_1 < \frac{n}{4}$.
  We now make a case distinction. If $C(L_0) \ge \frac{n}{2}$, then
  \[
    (C'(L_0) + D'(L_0)) \max\{1, \log D'(L_0)\}) \ge \frac{n}{2}.
  \]
  Otherwise, we consider all clusters of size $\ge L_0$.
  Initially, there were $D(L_0)$ such clusters, of total size $n - C(L_0) \ge \frac{n}{2}$.
  We deleted one cluster of size $L_1$,
  plus some clusters of total size at most $\frac{n}{4}$.
  This is at most half of $\frac{n}{2}$, and we always deleted a largest cluster,
  so we deleted at most $1 + \frac{1}{2}D(L_0)$ clusters in total.
  This shows $D'(L_0) \ge \frac{1}{2} D(L_0) - 1$. We also have $D'(L_0) \ge 1$ since $C'(L_0) < n'$. Thus,
  \[
    (C'(L_0) + D'(L_0)) \max\{1, \log D'(L_0)\}) \ge \frac{1}{8} (C(L_0) + D(L_0)) \max\{1, \log D(L_0)\}. \qedhere
  \]
\end{proof}

\begin{proposition} \label{prop:lower_block_sorting}
    Let $G \in \bG$ and define $M := \frac{1}{1000} \min \Big\{n, \min_{L \ge 1}^{C(L) < n/2} (C(L) + D(L)) \max\{1, \log D(L)\} \Big\}$.
    Then every algorithm $\cA$ satisfies
    \[
        \universal(\cA, G) \ge M.
    \]
\end{proposition}
\begin{proof}
    If $G$ contains only one cluster, then the proposition holds since $M < 1$, while any algorithm needs to do at least one comparison.
    Otherwise, let $M' = \min(n'/8, \frac{1}{32} \min_{L \ge 1}^{C'(L) < n'} (C'(L) + D'(L)) \max\{1, \log D'(L)\})$.
    For any algorithm $\cA$, the adversary plays the game for $M$ rounds.
    By \cref{lemm:linear_subset}, $M' \ge M$, so the total depth of all elements is at most $2 M'$.
    Thus, by \cref{clai:xmas_magic}, there exists a graph $G''$ on $X$ that is isomorphic to $G$ such that,
    for each cluster in $G''$, all elements share a root-to-leaf path in $T$.
    Let $U$ be the set of leaves of $T$.
    By pushing the elements in the same cluster down to a distinct leaf,
    the adversary constructs an input $I$ such that $G(I) = G'' \cong G$
    and all answers given by it are consistent with $I$.
    It follows that, on this input, $\cA$ did not find a duplicate after $M$ comparisons, even though there is a duplicate.
\end{proof}

\begin{theorem} \label{theo:lower_bound}
  Let $G \in \bG$. Then every algorithm $\cA$ satisfies
  \[
    \universal(\cA, G) \ge \frac{1}{1000} \min\Big(n + \min_{\substack{L \ge 2\\C(L) < n}} C(L) \log \frac{C(L)}{L}, \min_{\substack{L \ge 1\\C(L) < n/2}} (C(L) + D(L)) \max\{1, \log D(L)\}\Big).
  \]
\end{theorem}
\begin{proof}
    Combine \cref{theo:lower_median_recursion,prop:lower_block_sorting}.
\end{proof}

By combining \cref{theo:algo_oblivious} with \cref{theo:lower_bound}, we get:
\begin{theorem}
  There is an algorithm $\cA$ for element distinctness that is universally $O(\log \log n)$-competitive,
  that is, there is a constant $c > 0$ such that for every $G \in \bG$ and all algorithm $\cA'$, we have
  \[
    \universal(\cA, G) \le c \max\{1, \log \log (n)\} \cdot \universal(\cA', G).
  \]
\end{theorem}
This theorem is tight by \cref{theo:lower_competitive}.

\section{Preprocessing Algorithms for Element distinctness}

First, we consider a variant of universal optimality that was proposed by Cardinal et al.~\cite[STOC'10]{cardinal_sorting_2010} who consider the following: ``If one is allowed to preprocess the universally fixed graph $G$, can one then design competitive algorithms?''

Formally, we consider a two-stage algorithmic problem where the algorithm is first given a graph $G$ to preprocess.
Then, the algorithm gets the list $I$ of $n$ elements from some ordered set $U$,
with the promise that $G(I) \cong I$. From this point on, the algorithm can do comparisons between elements in $I$.
In this model, we want to minimize both the processing time and the number of comparisons performed.
To avoid trivial lower bounds from reading a dense graph $G$, we encode the graph with $m$ clusters as a sequence $s_1, \dots, s_m$ of cluster sizes where $s_1 + \dots + s_m = n$.

By \cref{theo:lower_bound}, even with unbounded preprocessing time,
any algorithm needs to perform at least
\[
\frac{1}{1000} \min\Big(n + \min_{\substack{L \ge 2\\C(L) < n}} C(L) \log \frac{C(L)}{L}, \min_{\substack{L \ge 1\\C(L) < n/2}} (C(L) + D(L)) \max\{1, \log D(L)\}\Big) \quad \text{comparisons}.
\]

We show that, after $O(m)$ preprocessing, we can give an $O(1)$-competitive algorithm.

\begin{theorem} \label{theo:algo_preprocessing}
    There is an algorithm that can preprocess any $G \in \bG$
    encoded as as $s_1, \dots, s_m$ in $O(m)$ time,
    and then, given an input $I$ with $G(I) \cong G$, can find a duplicate in 
    \[
    O\Big(\min\Big(n + \min_{\substack{L \ge 2\\C(L) < n}} C(L) \log \frac{C(L)}{L}, \min_{\substack{L \ge 1\\C(L) < n/2}} (C(L) + D(L)) \max\{1, \log D(L)\}\Big)\Big) \quad \text{comparisons}.
    \]
\end{theorem}
\begin{proof}
    Let
    \[
    L_1 = \arg\min_{\substack{L \ge 2\\C(L) < n}} C(L) \log \frac{C(L)}{L}, \quad L_2 =  \arg\min_{\substack{L \ge 1\\C(L) < n/2}} (C(L) + D(L)) \max\{1, \log D(L)\}.
    \]
    We first describe such an algorithm with $O(n m)$ preprocessing time:
    For each $L \in [n]$, compute $C(L)$ and $D(L)$ in $O(m)$ time.
    Determine $L_1$ and $L_2$.
    Then, given $I$, run either median recursion with $L = L_1$, or block sorting with $k = 2 D(L_2)$.
    By \cref{lemm:block_sorting,lemm:median_recursion}, this yields the number-of-comparisons bound in \cref{theo:algo_preprocessing}.

    To reduce the preprocessing time to $O(m \log m)$, observe that $C(L)$ and $D(L)$ are step functions that
    only change when $L = s_i$ for some $i \in m$. Hence, there are $O(m)$ possible values for $L_1$ and $L_2$.
    We can sort $s_1, \dots, s_m$ in $O(m \log m)$ time, and
    compute $C(L), D(L)$ for the relevant values of $L$ via sweep line in $O(m)$ time.
    Alternatively, we can achieve a preprocessing time of $O(n)$ by bucket sorting $s_1, \dots, s_m$.
    (This works even on a pointer machine, since $s_1 + \dots + s_m = n$.)

    We can further reduce the preprocessing time to $O(m)$.
    If $(C(L_2) + D(L_2)) \max\{1, \log D(L_2)\} \in \Omega(n)$,
    then the algorithm has to do $\Omega(n)$ comparisons,
    so we skip doing any preprocessing and compute $C(.), D(.), L_1, L_2$ in $O(n)$ time once we get access to $I$.
    Otherwise, $L_1$ is irrelevant (due to the $O(n + \dots)$ term) and we only have to consider $L_2$.
    We observe that we do not need to compute $L_2$ exactly.
    In fact, it suffices to approximate $\max\{1, \log D(L_2)\}$ up to a constant factor.
    Consider the following recursive procedure that operates on some list $S'$,
    where initially $S' = s_1, \dots, s_m$:
    Recursively find the median of $S'$, partition at it,
    and recurse on the half of $S'$ that contains values at least as big as the median.
    For $j = 1, \dots, \lceil \log m \rceil$,
    let $t_j$ be the minimum value of $S'$ in the $j$-th recursive call.
    During the recursion, for each $j$, compute $C(t_j)$ and $D(T_j)$. Afterwards, put
    \[
    L_2' = \arg \min_{j=1}^{\lceil \log m \rceil} \big((C(t_j) + D(t_j)) \max\{1, \log D(T_j)\}\big).
    \]
    The recursive procedure takes $O(m)$ time in total (even on a pointer machine).
    Suppose that the $\ell$-th recursive calls is the final recursive call where $L_2 \in S'$.
    Then $t_\ell \le L_2$ and $D(t_\ell) \le 2 D(L_2) + 1$.
    This implies
    \[
    \begin{aligned}
    (C(L_2') + D(L_2')) \max\{1, \log D(L_2')\} &\le (C(t_\ell) + D(t_\ell)) \max\{1, \log D(t_\ell)\} \\
    &\le 3 (C(L_2) + D(L_2)) \max\{1, \log D(L_2)\}.
    \end{aligned}
    \]
    Thus, in this case, running block recursion with $k = 2 D(L_2')$ yields the number-of-comparisons bound in \cref{theo:algo_preprocessing}.
\end{proof}

\section{Order Optimality for Element distinctness}

The main body of this paper considers universal optimality where the fixed object is an undirected graph $G$. 
We observed that each input $I$ can induce an undirected graph $G(I)$ where the vertex set is $\{1, 2, \ldots, n \}$ and there is an edge between $i$ and $j$ if and only if $I[i] = I[j]$. For every fixed undirected graph $G \in \bG_n$, we were interested in an algorithm's worst-case over all inputs $I$ where $G(I) \cong I$. 

However, we noted in the introduction that that previous papers~\cite{afshani2009instance,Hoog2024Tight,Haeupler2025Sorting} instead consider fixing, for (partially) ordered input $I$, the directed acyclic graph $\overrightarrow{G}(I)$.
This is a directed graph where the vertex set is $\{1, 2, \ldots, n \}$ and there is an edge from $i$ to $j$ if and only if $I[i] < I[j]$.
Note that the undirected complement of $\overrightarrow{G}(I)$ is $G(I)$.

We observe that universally optimality can also be defined by fixing a such DAG. 
Formally, denote by $\overrightarrow{\bG}_n$ the set of all $n$-vertex DAGs $\overrightarrow{G}$ where the undirected complement of $\overrightarrow{G}$  is in $\bG$. 
Then, given a fixed $\overrightarrow{G} \in \overrightarrow{\bG}_n$, we can alternatively define the universal running time of an algorithm as the worst-case running time over all inputs $I$ such that $\overrightarrow{G}(I) \cong G$.
To distinguish between these two concepts, we use a different notation for the latter notion of optimality: 

\[
  \textnormal{Let } \overrightarrow{G} \in \overrightarrow{\bG}_n \textnormal{. We define: } \quad \daguniversal(\cA,  \overrightarrow{G}) = \max_{I : \overrightarrow{G}(I) \cong \overrightarrow{G}} \runtime(\cA, I).
\]

Observe that the above is equivalent to fixing $I$ and considering an algorithm's worst-case running time over all input lists that are a permutation of $I$ (thus, it is equivalent to the notion of \emph{instance optimality in the order-oblivious setting} in~\cite{afshani2009instance}). 
Recall that, to distinguish between our two models of optimality, we say that an algorithm is $O(f(n))$-\emph{order-competitive} if there exists a function $g \in O(f)$ such that, for all sufficiently large $n$, for all graphs $\overrightarrow{G} \in \overrightarrow{\bG}_n$, and for all algorithms $\cA'$, we have
\[
\textnormal{Order}(\cA,\overrightarrow{G}) \leq g(n) \cdot \textnormal{Order}(\cA',\overrightarrow{G}).
\]

We show that a simple doubling algorithm achieves tight order-competitive bounds. 
Consider the following algorithm $\cA$:
Starting from $k=2$, pick $\min(n, k)$ arbitrary elements and sort them.
If there is a pair of equal elements, output them, otherwise double $k$ and repeat.
\begin{lemma} \label{lemm:upper_dag_universal}
    The algorithm $\cA$ is $O(\log n)$-order-competitive,
    that is, we prove that there exists a constant $c>0$ such for every $n$-vertex DAG $\overrightarrow{G} \in \overrightarrow{\bG}$ and every algorithm $\cA'$, we have $    \daguniversal(\cA, \overrightarrow{G}) \le c \max\{1, \log(n)\} \cdot \daguniversal(\cA', \overrightarrow{G}).$
\end{lemma}
\begin{proof}
    Let $\cA'$ and $\overrightarrow{G}$ be arbitrary.
    Let $k$ be the largest power of two for which there is an input $I$ with $\overrightarrow{G}(I) \cong \overrightarrow{G}$ so that $\cA$, when run on $I$, does not find a duplicate for this value of $k$.
    Then $\cA$ always finds a duplicate when sorting at most $2k$ elements,
    so $\daguniversal(\cA', \overrightarrow{G}) \in O(k \log k) \subseteq O(k \log n)$.
    On the other hand, since $\cA$ did not find a duplicate,
    there are $k$ distinct elements in $I$.
    There is a permutation $I'$ of $I$ such that
    the first $k/2$ comparisons by $\cA'$ happen among these $k$ elements.
    Then, $\overrightarrow{G}(I') \cong \overrightarrow{G}(I)$ and $\cA'$ does not find a duplicate in $I'$ after $k/2$ comparisons,
    so $k/2 \le \daguniversal(\cA', \overrightarrow{G})$.
\end{proof}
We show that \cref{lemm:upper_dag_universal} is tight: no algorithm is universally $o(\log n)$-competitive.

\begin{lemma}
  For every algorithm $\cA$ and every positive integer $n$, there is a DAG $\ora{G}$ and an algorithm $\cA'$
  such that $\daguniversal(\cA, \ora{G}) \in \Omega(n \log n)$ but $\daguniversal(\cA', \ora{G}) \in O(n)$.
\end{lemma}
\begin{proof}
    Let $\cA$ and $n$ be arbitrary. Let $G$ be an undirected graph with $1$ cluster of size $2$ and $n-2$ clusters of size $1$.
    By \cref{theo:lower_median_recursion}, there is an input $I$ with $G(I) \cong G$ and $\runtime(\cA, I) \in \Omega(n \log n)$.
    Let $\ora{G} = \ora{G}(I)$ and
    let $k$ be the rank in $\ora{G}$ of the size-2 cluster.
    On the input $I$, the algorithm $\cA'$ runs quick select twice, with parameter $k$ and $k+1$, respectively.
    This takes $O(n)$ time and yields the elements in $I$ with rank $k$ and $k+1$, respectively.
    Then, $\cA'$ outputs these two elements.
\end{proof}

\section{Set Intersection}
We finally consider universal optimality and order optimality for the set intersection problem.
Recall that the input is two lists $A$ and $B$, each of size $n$,
and that $G(A, B)$ is a graph with $2n$ vertices.
We use \emph{$A$-vertices} and \emph{$B$-vertices} to refer to the $n$ vertices that correspond to elements in $A$
or $B$, respectively.
The connected components of $G(A, B)$ are cliques, which we call \emph{clusters}.
The \emph{$A$-size} and \emph{$B$-size} of a cluster is the number of $A$-vertices or $B$-vertices, respectively.
We denote by $\bG_{n, n}$ the set of all such graphs with such a partition into $A$-vertices and $B$-vertices.
Two graphs in $\bG_{n, n}$ are isomorphic if there is an isomorphism of undirected graphs that sends $A$-vertices
to $A$-vertices and $B$-vertices to $B$-vertices.
\[
  \text{For fixed }G \in \bG_{n, n},\text{ we define:} \quad \universal(\cA,G) = \max_{(A, B):\, G(A, B)\cong G}\; \runtime(\cA,A, B).
\]
We denote by $\ora{G}_{n, n}$ the set of all $2n$-vertex DAGs, with a partition into $A$-vertices and $B$-vertices, whose undirected complement lies $\bG_{n, n}$.
\[
    \text{For fixed }\ora{G} \in \ora{\bG}_{n, n},\text{ we define:} \quad \daguniversal(\cA,\ora{G}) = \max_{(A, B):\, \ora{G}(A, B)\cong \ora{G}}\; \runtime(\cA,A, B).
\]

We show that a simple doubling algorithm achieves tight bounds for the set intersection problem.
The algorithm $\cA$ works as follows:
Starting from $k=2$, pick $\min(|A|, k)$ elements from $A$ and $\min(|B|, k)$ from $B$ and sort them.
If there is a pair of equal elements from $A$ and $B$, output them, otherwise double $k$ and repeat.
\begin{lemma} \label{lemm:upper_set_intersection}
    The algorithm $\cA$ is $O(\log n)$-competitive,
    that is, there is a function $f \in O(\log n)$ such for all $G \in \bG_n$ and for every algorithm $\cA'$, we have
    \[
        \universal(\cA, G) \le f(n) \cdot \universal(\cA', G).
    \]
\end{lemma}
\begin{proof}
    The proof closely resembles that of \cref{lemm:upper_dag_universal}.
    Let $\cA'$ and $G$ be arbitrary.
    Let $k$ be the largest power of two for which there is an input $(A, B)$ with $G(A, B) \cong G$ so that $\cA$, when run on $(A, B)$, does not find a set intersection for this value of $k$.
    Then $k \le |A| + |B|$, and $\cA$ always finds a set intersection
    when sorting at most $4k$ elements,
    so $\universal(\cA', G) \in O(k \log k) \subseteq O(k \log n)$.
    On the other hand, since $\cA$ did not find a set intersection,
    there are $\min(|A|, k)$ from $A$ and $\min(|B|, k)$ elements from $B$ with no set intersection. Since $|A| + |B| \ge k$, these are $\ge k$ elements in total.
    There is a permutation $A'$ of $A$ and a permutation $B'$ of $B$ such that
    the first $k/2$ comparisons by $\cA'$ happen among these $\ge k$ elements.
    Then, $G(A', B') \cong G(A, B)$ and $\cA'$ does not find a set intersection in $(A', B')$ after $k/2$ comparisons,
    so $k/2 \le \universal(\cA', G)$.
\end{proof}

\begin{lemma} \label{lemm:upper_order_set_intersection}
    The algorithm $\cA$ is also $O(\log n)$-order competitive.
    that is, there is a function $f \in O(\log n)$ such for all $\ora{G} \in \ora{\bG}_n$ and for every algorithm $\cA'$, we have
    \[
        \daguniversal(\cA, \ora{G}) \le f(n) \cdot \daguniversal(\cA', \ora{G}).
    \]
\end{lemma}
\begin{proof}
    Let $\cA'$ and $\ora{G}$ be arbitrary.
    Let $k$ be the largest power of two for which there is an input $(A, B)$ with $\ora{G}(A, B) \cong \ora{G}$ so that $\cA$, when run on $(A, B)$, does not find a set intersection for this value of $k$.
    Then, proceed as in the proof of \cref{lemm:upper_set_intersection}.
\end{proof}
We show that \cref{lemm:upper_set_intersection,lemm:upper_order_set_intersection} are tight: no algorithm is universally $o(\log n)$-competitive, or universally $o(\log n)$-order competitive.

\begin{lemma}
  For every algorithm $\cA$ and every positive integer $n$, there is a graph $G \in \bG_{n, n}$ and algorithm $\cA'$
  such that $\universal(\cA, G) \in \Omega(n \log n)$ but $\universal(\cA', G) \in O(n)$.
  Moreover, there is is a DAG $\ora{G}$ whose undirected complement is $G$ such that
  such that $\daguniversal(\cA, \ora{G}) \in \Omega(n \log n)$ but $\daguniversal(\cA', \ora{G}) \in O(n)$.
\end{lemma}
\begin{proof}
    We first show the bounds on the universal running time.
    Let $\cA$ and $n$ be arbitrary.
    For $i \in [n^{1/3}]$, consider the graph $G_i \in \bG_{n, n}$ with:
    \begin{enumerate}[(1)]
        \item For $j \in [n^{1/3}]$, one cluster of $A$-size $j$ and $B$-size $\delta_{i, j}$, that is, $1$ if $i = j$ and $0$ otherwise.
        \item $n-1$ clusters of $A$-size $0$ and $B$-size $1$.
        \item One cluster of $A$-size $n - \Theta(n^{2/3})$ and $B$-size $0$.
    \end{enumerate}
    Note that the cluster of type (1) with $A$-size $i$ is the only cluster that contains both elements from $A$ and from $B$.
    The algorithm $\cA_i$ operates as follows:
    First compute the median of $A$ to find the cluster (3) and delete the corresponding elements from $A$.
    Then, sort the remaining $O(n^{2/3})$ elements of $A$.
    This partitions $A$ into the $n^{1/3}$ clusters of type (1).
    Among these, find the cluster of $A$-size $i$.
    Pick an arbitrary element from this cluster, and compare it to all elements in $B$ to find the set intersection.
    We have $\universal(\cA_i, G_i) \in O(n + n^{2/3} \log n) = O(n)$.

    Conversely, let $\cA$ be any algorithm.
    Without loss of generality, suppose $n^{1/3} = 2^{\ell}$ for some integer $\ell$.
    We consider an adaptive adversary that uses a modified version of the binary tree strategy:
    Consider a complete binary tree $T$ of height $n \log n$.
    The elements of $B$ all start out at the root of $T$.
    The elements of $A$ all start at a leaf of $T$:
    The elements in the cluster (3) start out at the leftmost leaf of $T$.
    (These elements are smaller than all other elements, and thus irrelevant for the rest of the proof.)
    For each cluster of type (1), pick a different node $u$ of depth $\ell$.
    The $A$-elements in this cluster start at the same arbitrary leaf in the subtree of $u$.
    The adversary answers comparisons by moving elements in $B$ down the tree,
    as described in Section~\ref{sect:no_competitive}.

    We play the adversarial game with $\cA$ for $\frac{1}{2} n \cdot \ell = \frac{1}{6} n \log n$ rounds.
    Then, there is an element $x \in A$ that is at a tree node of depth $\le \ell$.
    In particular, in the subtree of $x$, there is a leaf that contains a cluster of type (1), say the one with index $j$.
    We add $x$ to this cluster, and turn all other elements of $A$ into clusters of size 1.
    This defines a graph $G$ on $X$ with $G \cong G_j$, where each cluster consists of elements on the same root-to-leaf path of $T$.
    Thus, there is an input $(A, B)$ with $G(A, B) \cong G_j$ such that the answers given by the adversary are consistent with $(A, B)$.
    It follows that $\cA$ did not find a set intersection after $\frac{1}{6} n \log n$ comparisons,
    even though there is one. Therefore, $\universal(\cA, G_j) \ge \frac{1}{6} n \log n$.
    We put $\cA' = \cA_j$, then $\universal(\cA', G_j) \in O(n)$ as noted earlier. This shows the universal optimality part of Lemma.

    Finally, we show the bounds on the order running time. Put $\ora{G} = \ora{G}(A, B)$.
    Then, $\daguniversal(\cA, \ora{G}) \in \Omega(n \log n)$ since $\cA$ does at least $\frac{1}{6} n \log n$ comparisons on $(A, B)$.
    On the other hand, $\daguniversal(\cA', \ora{G}) \le \universal(\cA', \ora{G}) \in O(n)$,
    since we take the maximum over a smaller subset of inputs. This shows the Lemma.
\end{proof}


\bibliographystyle{ACM-Reference-Format}
\bibliography{refs}


\begin{thebibliography}{28}


\ifx \showCODEN    \undefined \def \showCODEN     #1{\unskip}     \fi
\ifx \showISBNx    \undefined \def \showISBNx     #1{\unskip}     \fi
\ifx \showISBNxiii \undefined \def \showISBNxiii  #1{\unskip}     \fi
\ifx \showISSN     \undefined \def \showISSN      #1{\unskip}     \fi
\ifx \showLCCN     \undefined \def \showLCCN      #1{\unskip}     \fi
\ifx \shownote     \undefined \def \shownote      #1{#1}          \fi
\ifx \showarticletitle \undefined \def \showarticletitle #1{#1}   \fi
\ifx \showURL      \undefined \def \showURL       {\relax}        \fi
\providecommand\bibfield[2]{#2}
\providecommand\bibinfo[2]{#2}
\providecommand\natexlab[1]{#1}
\providecommand\showeprint[2][]{arXiv:#2}

\bibitem[Afshani et~al\mbox{.}(2009)]%
        {afshani2009instance}
\bibfield{author}{\bibinfo{person}{Peyman Afshani}, \bibinfo{person}{J{\'e}r{\'e}my Barbay}, {and} \bibinfo{person}{Timothy Chan}.} \bibinfo{year}{2009}\natexlab{}.
\newblock \showarticletitle{Instance-Optimal Geometric Algorithms}. In \bibinfo{booktitle}{\emph{IEEE Symposium on Foundations of Computer Science (FOCS)}}. \bibinfo{publisher}{Association for Computing Machinery}, \bibinfo{pages}{129--138}.
\newblock
\href{https://doi.org/10.1145/3046673}{doi:\nolinkurl{10.1145/3046673}}


\bibitem[Ajtai(1988)]%
        {Ajtai1988}
\bibfield{author}{\bibinfo{person}{Mikl{\'o}s Ajtai}.} \bibinfo{year}{1988}\natexlab{}.
\newblock \showarticletitle{A lower bound for finding predecessors in {Yao}'s cell probe model}.
\newblock \bibinfo{journal}{\emph{Combinatorica}} \bibinfo{volume}{8}, \bibinfo{number}{3} (\bibinfo{year}{1988}), \bibinfo{pages}{235--247}.
\newblock
\href{https://doi.org/10.1007/BF02126786}{doi:\nolinkurl{10.1007/BF02126786}}


\bibitem[Beame et~al\mbox{.}(1991)]%
        {BeameBorodinSaksSkyum1991}
\bibfield{author}{\bibinfo{person}{Paul Beame}, \bibinfo{person}{Allan Borodin}, \bibinfo{person}{Michael Saks}, {and} \bibinfo{person}{Sven Skyum}.} \bibinfo{year}{1991}\natexlab{}.
\newblock \showarticletitle{Time–Space Tradeoffs for Branching Programs}.
\newblock \bibinfo{journal}{\emph{J. Comput. System Sci.}} \bibinfo{volume}{44}, \bibinfo{number}{2} (\bibinfo{year}{1991}), \bibinfo{pages}{272--298}.
\newblock
\href{https://doi.org/10.1016/0022-0000(91)90035-2}{doi:\nolinkurl{10.1016/0022-0000(91)90035-2}}


\bibitem[Beame et~al\mbox{.}(2013)]%
        {Beame2013FOCS}
\bibfield{author}{\bibinfo{person}{Paul Beame}, \bibinfo{person}{Raphael Clifford}, {and} \bibinfo{person}{Widad Machmouchi}.} \bibinfo{year}{2013}\natexlab{}.
\newblock \showarticletitle{Element Distinctness, Frequency Moments, and Sliding Windows}. In \bibinfo{booktitle}{\emph{Symposium on Foundations of Computer Science (FOCS)}} \emph{(\bibinfo{series}{FOCS '13})}. \bibinfo{publisher}{IEEE Computer Society}, \bibinfo{address}{USA}, \bibinfo{pages}{290–299}.
\newblock
\showISBNx{9780769551357}
\href{https://doi.org/10.1109/FOCS.2013.39}{doi:\nolinkurl{10.1109/FOCS.2013.39}}


\bibitem[Cardinal et~al\mbox{.}(2010)]%
        {cardinal_sorting_2010}
\bibfield{author}{\bibinfo{person}{Jean Cardinal}, \bibinfo{person}{Samuel Fiorini}, \bibinfo{person}{Gwenaël Joret}, \bibinfo{person}{Raphaël~M. Jungers}, {and} \bibinfo{person}{J.~Ian Munro}.} \bibinfo{year}{2010}\natexlab{}.
\newblock \showarticletitle{Sorting under partial information (without the ellipsoid algorithm)}. In \bibinfo{booktitle}{\emph{{ACM} Symposium on {Theory} of Computing (STOC)}} \emph{(\bibinfo{series}{{STOC} '10})}. \bibinfo{publisher}{Association for Computing Machinery}, \bibinfo{address}{New York, NY, USA}.
\newblock
\showISBNx{978-1-4503-0050-6}
\href{https://doi.org/10.1145/1806689.1806740}{doi:\nolinkurl{10.1145/1806689.1806740}}


\bibitem[Chakrabarti et~al\mbox{.}(2008)]%
        {Chakrabarti2008}
\bibfield{author}{\bibinfo{person}{Amit Chakrabarti}, \bibinfo{person}{Graham Cormode}, {and} \bibinfo{person}{Andrew McGregor}.} \bibinfo{year}{2008}\natexlab{}.
\newblock \showarticletitle{Robust lower bounds for communication and stream computation}. In \bibinfo{booktitle}{\emph{Symposium on Theory of Computing (STOC)}} (Victoria, British Columbia, Canada) \emph{(\bibinfo{series}{STOC '08})}. \bibinfo{publisher}{Association for Computing Machinery}, \bibinfo{address}{New York, NY, USA}, \bibinfo{pages}{641–650}.
\newblock
\showISBNx{9781605580470}
\href{https://doi.org/10.1145/1374376.1374470}{doi:\nolinkurl{10.1145/1374376.1374470}}


\bibitem[Demaine et~al\mbox{.}(2004)]%
        {Demaine2004Dynamic}
\bibfield{author}{\bibinfo{person}{Erik~D. Demaine}, \bibinfo{person}{Dion Harmon}, \bibinfo{person}{John Iacono}, {and} \bibinfo{person}{Mihai Patrascu}.} \bibinfo{year}{2004}\natexlab{}.
\newblock \showarticletitle{Dynamic Optimality " Almost}. In \bibinfo{booktitle}{\emph{IEEE Symposium on Foundations of Computer Science (FOCS)}}. \bibinfo{publisher}{IEEE Computer Society}, \bibinfo{address}{USA}, \bibinfo{pages}{484–490}.
\newblock
\showISBNx{0769522289}
\href{https://doi.org/10.1109/FOCS.2004.23}{doi:\nolinkurl{10.1109/FOCS.2004.23}}


\bibitem[Fredman(1976)]%
        {fredman_how_1976}
\bibfield{author}{\bibinfo{person}{Michael~L. Fredman}.} \bibinfo{year}{1976}\natexlab{}.
\newblock \showarticletitle{How good is the information theory bound in sorting?}
\newblock \bibinfo{journal}{\emph{Theoretical Computer Science}} \bibinfo{volume}{1}, \bibinfo{number}{4} (\bibinfo{date}{April} \bibinfo{year}{1976}), \bibinfo{pages}{355--361}.
\newblock
\showISSN{0304-3975}
\href{https://doi.org/10.1016/0304-3975(76)90078-5}{doi:\nolinkurl{10.1016/0304-3975(76)90078-5}}


\bibitem[Haeupler et~al\mbox{.}(2024)]%
        {Haeupler2024Dijkstra}
\bibfield{author}{\bibinfo{person}{Bernhard Haeupler}, \bibinfo{person}{Richard Hlad{\'{\i}}k}, \bibinfo{person}{V{\'{a}}clav Rozhon}, \bibinfo{person}{Robert~E. Tarjan}, {and} \bibinfo{person}{Jakub Tetek}.} \bibinfo{year}{2024}\natexlab{}.
\newblock \showarticletitle{Universal Optimality of Dijkstra Via Beyond-Worst-Case Heaps}. In \bibinfo{booktitle}{\emph{Annual Symposium on Foundations of Computer Science ({FOCS})}}. \bibinfo{publisher}{{IEEE}}, \bibinfo{pages}{2099--2130}.
\newblock
\urldef\tempurl%
\url{https://doi.org/10.1109/FOCS61266.2024.00125}
\showURL{%
\tempurl}


\bibitem[Haeupler et~al\mbox{.}(2025)]%
        {Haeupler2025Sorting}
\bibfield{author}{\bibinfo{person}{Bernhard Haeupler}, \bibinfo{person}{Richard Hladík}, \bibinfo{person}{John Iacono}, \bibinfo{person}{Václav Rozhoň}, \bibinfo{person}{Robert~E. Tarjan}, {and} \bibinfo{person}{Jakub Tětek}.} \bibinfo{year}{2025}\natexlab{}.
\newblock \showarticletitle{Fast and Simple Sorting Using Partial Information}. In \bibinfo{booktitle}{\emph{Annual ACM-SIAM Symposium on Discrete Algorithms (SODA)}}. \bibinfo{pages}{3953--3973}.
\newblock
\href{https://doi.org/10.1137/1.9781611978322.134}{doi:\nolinkurl{10.1137/1.9781611978322.134}}


\bibitem[Haeupler et~al\mbox{.}(2021)]%
        {Haeupler2021UniversalDistributed}
\bibfield{author}{\bibinfo{person}{Bernhard Haeupler}, \bibinfo{person}{David Wajc}, {and} \bibinfo{person}{Goran Zuzic}.} \bibinfo{year}{2021}\natexlab{}.
\newblock \showarticletitle{Universally-optimal distributed algorithms for known topologies}. In \bibinfo{booktitle}{\emph{ACM SIGACT Symposium on Theory of Computing (STOC)}}. \bibinfo{publisher}{Association for Computing Machinery}, \bibinfo{address}{New York, NY, USA}, \bibinfo{pages}{1166–1179}.
\newblock
\showISBNx{9781450380539}
\href{https://doi.org/10.1145/3406325.3451081}{doi:\nolinkurl{10.1145/3406325.3451081}}


\bibitem[Han(2002)]%
        {Han2002Deterministic}
\bibfield{author}{\bibinfo{person}{Yijie Han}.} \bibinfo{year}{2002}\natexlab{}.
\newblock \showarticletitle{Deterministic sorting in O(n log log n) time and linear space}. In \bibinfo{booktitle}{\emph{ACM Symposium on Theory of Computing (STOC)}} (Montreal, Quebec, Canada) \emph{(\bibinfo{series}{STOC '02})}. \bibinfo{publisher}{Association for Computing Machinery}, \bibinfo{address}{New York, NY, USA}, \bibinfo{pages}{602–608}.
\newblock
\showISBNx{1581134959}
\href{https://doi.org/10.1145/509907.509993}{doi:\nolinkurl{10.1145/509907.509993}}


\bibitem[Han and Thorup(2002)]%
        {Han2002Random}
\bibfield{author}{\bibinfo{person}{Yijie Han} {and} \bibinfo{person}{Mikkel Thorup}.} \bibinfo{year}{2002}\natexlab{}.
\newblock \showarticletitle{Integer Sorting in 0(n sqrt (log log n)) Expected Time and Linear Space}. In \bibinfo{booktitle}{\emph{Foundations of Computer Science ({FOCS})}}. \bibinfo{publisher}{{IEEE} Computer Society}, \bibinfo{pages}{135--144}.
\newblock
\href{https://doi.org/10.1109/SFCS.2002.1181890}{doi:\nolinkurl{10.1109/SFCS.2002.1181890}}


\bibitem[Hlad{\'\i}k and T\v{e}tek(2025)]%
        {hladik2025FORC}
\bibfield{author}{\bibinfo{person}{Richard Hlad{\'\i}k} {and} \bibinfo{person}{Jakub T\v{e}tek}.} \bibinfo{year}{2025}\natexlab{}.
\newblock \showarticletitle{{Near-Universally-Optimal Differentially Private Minimum Spanning Trees}}. In \bibinfo{booktitle}{\emph{Foundations of Responsible Computing (FORC)}} \emph{(\bibinfo{series}{Leibniz International Proceedings in Informatics (LIPIcs)}, Vol.~\bibinfo{volume}{329})}, \bibfield{editor}{\bibinfo{person}{Mark Bun}} (Ed.). \bibinfo{publisher}{Schloss Dagstuhl -- Leibniz-Zentrum f{\"u}r Informatik}, \bibinfo{address}{Dagstuhl, Germany}, \bibinfo{pages}{6:1--6:19}.
\newblock
\showISBNx{978-3-95977-367-6}
\showISSN{1868-8969}
\href{https://doi.org/10.4230/LIPIcs.FORC.2025.6}{doi:\nolinkurl{10.4230/LIPIcs.FORC.2025.6}}


\bibitem[Kahn and Kim(1992)]%
        {kahn_entropy_1992}
\bibfield{author}{\bibinfo{person}{Jeff Kahn} {and} \bibinfo{person}{Jeong~Han Kim}.} \bibinfo{year}{1992}\natexlab{}.
\newblock \showarticletitle{Entropy and sorting}. In \bibinfo{booktitle}{\emph{Proceedings of the twenty-fourth annual {ACM} symposium on {Theory} of {Computing}}} \emph{(\bibinfo{series}{{STOC} '92})}. \bibinfo{publisher}{Association for Computing Machinery}, \bibinfo{address}{New York, NY, USA}, \bibinfo{pages}{178--187}.
\newblock
\showISBNx{978-0-89791-511-3}
\href{https://doi.org/10.1145/129712.129731}{doi:\nolinkurl{10.1145/129712.129731}}


\bibitem[Kahn and Saks(1984)]%
        {kahn_balancing_1984}
\bibfield{author}{\bibinfo{person}{Jeff Kahn} {and} \bibinfo{person}{Michael Saks}.} \bibinfo{year}{1984}\natexlab{}.
\newblock \showarticletitle{Balancing poset extensions}.
\newblock \bibinfo{journal}{\emph{Order}} \bibinfo{volume}{1}, \bibinfo{number}{2} (\bibinfo{date}{June} \bibinfo{year}{1984}), \bibinfo{pages}{113--126}.
\newblock
\showISSN{1572-9273}
\href{https://doi.org/10.1007/BF00565647}{doi:\nolinkurl{10.1007/BF00565647}}


\bibitem[Kirkpatrick and Seidel(1986)]%
        {kirkpatrick1986ultimate}
\bibfield{author}{\bibinfo{person}{David~G Kirkpatrick} {and} \bibinfo{person}{Raimund Seidel}.} \bibinfo{year}{1986}\natexlab{}.
\newblock \showarticletitle{The ultimate planar convex hull algorithm?}
\newblock \bibinfo{journal}{\emph{SIAM J. Comput.}} \bibinfo{volume}{15}, \bibinfo{number}{1} (\bibinfo{year}{1986}), \bibinfo{pages}{287--299}.
\newblock
\href{https://doi.org/10.1137/0215021}{doi:\nolinkurl{10.1137/0215021}}


\bibitem[Kopelowitz et~al\mbox{.}(2016)]%
        {Kopelowitz2016Higher}
\bibfield{author}{\bibinfo{person}{Tsvi Kopelowitz}, \bibinfo{person}{Seth Pettie}, {and} \bibinfo{person}{Ely Porat}.} \bibinfo{year}{2016}\natexlab{}.
\newblock \showarticletitle{Higher lower bounds from the 3SUM conjecture}. In \bibinfo{booktitle}{\emph{ACM-SIAM Symposium on Discrete Algorithms (SODA)}} (Arlington, Virginia) \emph{(\bibinfo{series}{SODA '16})}. \bibinfo{publisher}{Society for Industrial and Applied Mathematics}, \bibinfo{address}{USA}, \bibinfo{pages}{1272–1287}.
\newblock
\showISBNx{9781611974331}
\href{https://doi.org/10.5555/2884435.2884524}{doi:\nolinkurl{10.5555/2884435.2884524}}


\bibitem[Ostlin and Pagh(2003)]%
        {Ostlin2003Hashing}
\bibfield{author}{\bibinfo{person}{Anna Ostlin} {and} \bibinfo{person}{Rasmus Pagh}.} \bibinfo{year}{2003}\natexlab{}.
\newblock \showarticletitle{Uniform hashing in constant time and linear space}. In \bibinfo{booktitle}{\emph{ACM Symposium on Theory of Computing (STOC)}} (San Diego, CA, USA) \emph{(\bibinfo{series}{STOC '03})}. \bibinfo{publisher}{Association for Computing Machinery}, \bibinfo{address}{New York, NY, USA}, \bibinfo{pages}{622–628}.
\newblock
\showISBNx{1581136749}
\href{https://doi.org/10.1145/780542.780633}{doi:\nolinkurl{10.1145/780542.780633}}


\bibitem[Patt-Shamir and Peleg(1993)]%
        {pattshamir1993Space}
\bibfield{author}{\bibinfo{person}{Boaz Patt-Shamir} {and} \bibinfo{person}{David Peleg}.} \bibinfo{year}{1993}\natexlab{}.
\newblock \showarticletitle{Time-space tradeoffs for set operations}.
\newblock \bibinfo{journal}{\emph{Theoreticl Compututer Science}} \bibinfo{volume}{110}, \bibinfo{number}{1} (\bibinfo{date}{March} \bibinfo{year}{1993}), \bibinfo{pages}{99–129}.
\newblock
\showISSN{0304-3975}
\href{https://doi.org/10.1016/0304-3975(93)90352-T}{doi:\nolinkurl{10.1016/0304-3975(93)90352-T}}


\bibitem[P\u{a}tra\c{s}cu and Demaine(2006)]%
        {Patrascu}
\bibfield{author}{\bibinfo{person}{Mihai P\u{a}tra\c{s}cu} {and} \bibinfo{person}{Erik~D. Demaine}.} \bibinfo{year}{2006}\natexlab{}.
\newblock \showarticletitle{Logarithmic Lower Bounds in the Cell-Probe Model}.
\newblock \bibinfo{journal}{\emph{SIAM J. Comput.}} \bibinfo{volume}{35}, \bibinfo{number}{4} (\bibinfo{year}{2006}), \bibinfo{pages}{932--963}.
\newblock
\href{https://doi.org/10.1137/S0097539705447256}{doi:\nolinkurl{10.1137/S0097539705447256}}


\bibitem[Sleator and Tarjan(1985)]%
        {Sleator1985self}
\bibfield{author}{\bibinfo{person}{Daniel~Dominic Sleator} {and} \bibinfo{person}{Robert~Endre Tarjan}.} \bibinfo{year}{1985}\natexlab{}.
\newblock \showarticletitle{Self-adjusting binary search trees}.
\newblock \bibinfo{journal}{\emph{Journal of the ACM (JACM)}} \bibinfo{volume}{32}, \bibinfo{number}{3} (\bibinfo{date}{July} \bibinfo{year}{1985}), \bibinfo{pages}{652–686}.
\newblock
\showISSN{0004-5411}
\href{https://doi.org/10.1145/3828.3835}{doi:\nolinkurl{10.1145/3828.3835}}


\bibitem[{van der Hoog} et~al\mbox{.}(2025a)]%
        {Hoog2025Convex}
\bibfield{author}{\bibinfo{person}{Ivor {van der Hoog}}, \bibinfo{person}{Eva Rotenberg}, {and} \bibinfo{person}{Daniel Rutschmann}.} \bibinfo{year}{2025}\natexlab{a}.
\newblock \showarticletitle{A Combinatorial Proof of Universal Optimality for Computing a Planar Convex Hull}. In \bibinfo{booktitle}{\emph{European Symposium on Algorithms ({ESA})}} \emph{(\bibinfo{series}{LIPIcs}, Vol.~\bibinfo{volume}{351})}, \bibfield{editor}{\bibinfo{person}{Anne Benoit}, \bibinfo{person}{Haim Kaplan}, \bibinfo{person}{Sebastian Wild}, {and} \bibinfo{person}{Grzegorz Herman}} (Eds.). \bibinfo{publisher}{Schloss Dagstuhl - Leibniz-Zentrum f{\"{u}}r Informatik}, \bibinfo{pages}{102:1--102:13}.
\newblock
\href{https://doi.org/10.4230/LIPICS.ESA.2025.102}{doi:\nolinkurl{10.4230/LIPICS.ESA.2025.102}}


\bibitem[{van der Hoog} et~al\mbox{.}(2025b)]%
        {Hoog2025SimplerDAG}
\bibfield{author}{\bibinfo{person}{Ivor {van der Hoog}}, \bibinfo{person}{Eva Rotenberg}, {and} \bibinfo{person}{Daniel Rutschmann}.} \bibinfo{year}{2025}\natexlab{b}.
\newblock \showarticletitle{Simpler Optimal Sorting from a Directed Acyclic Graph}. In \bibinfo{booktitle}{\emph{Symposium on Simplicity in Algorithms (SOSA)}}. \bibinfo{publisher}{{SIAM}}.
\newblock
\href{https://doi.org/10.1137/1.9781611978315.26}{doi:\nolinkurl{10.1137/1.9781611978315.26}}


\bibitem[{van der Hoog} et~al\mbox{.}(2025c)]%
        {Hoog2025SimplerSorting}
\bibfield{author}{\bibinfo{person}{Ivor {van der Hoog}}, \bibinfo{person}{Eva Rotenberg}, {and} \bibinfo{person}{Daniel Rutschmann}.} \bibinfo{year}{2025}\natexlab{c}.
\newblock \showarticletitle{Simpler Universally Optimal Dijkstra}. In \bibinfo{booktitle}{\emph{European Symposium on Algorithms (ESA)}} \emph{(\bibinfo{series}{LIPIcs}, Vol.~\bibinfo{volume}{351})}. \bibinfo{publisher}{Schloss Dagstuhl - Leibniz-Zentrum f{\"{u}}r Informatik}, \bibinfo{pages}{71:1--71:9}.
\newblock
\urldef\tempurl%
\url{https://doi.org/10.4230/LIPIcs.ESA.2025.71}
\showURL{%
\tempurl}


\bibitem[{van der Hoog} and Rutschmann(2024)]%
        {Hoog2024Tight}
\bibfield{author}{\bibinfo{person}{Ivor {van der Hoog}} {and} \bibinfo{person}{Daniel Rutschmann}.} \bibinfo{year}{2024}\natexlab{}.
\newblock \showarticletitle{Tight Bounds for Sorting Under Partial Information}. In \bibinfo{booktitle}{\emph{Foundations of Computer Science, {FOCS} 2024, Chicago, IL, USA, October 27-30, 2024}}. \bibinfo{publisher}{{IEEE}}, \bibinfo{pages}{2243--2252}.
\newblock
\href{https://doi.org/10.1109/FOCS61266.2024.00131}{doi:\nolinkurl{10.1109/FOCS61266.2024.00131}}


\bibitem[Yao(1994)]%
        {YaoNear-Optimal}
\bibfield{author}{\bibinfo{person}{Andrew Chi-Chih Yao}.} \bibinfo{year}{1994}\natexlab{}.
\newblock \showarticletitle{Near-Optimal Time-Space Tradeoff for Element Distinctness}.
\newblock \bibinfo{journal}{\emph{SIAM J. Comput.}} \bibinfo{volume}{23}, \bibinfo{number}{5} (\bibinfo{year}{1994}), \bibinfo{pages}{966--975}.
\newblock
\href{https://doi.org/10.1137/S0097539788148959}{doi:\nolinkurl{10.1137/S0097539788148959}}


\bibitem[Zuzic et~al\mbox{.}(2022)]%
        {Zuzic2022Universally}
\bibfield{author}{\bibinfo{person}{Goran Zuzic}, \bibinfo{person}{Gramoz Goranci}, \bibinfo{person}{Mingquan Ye}, \bibinfo{person}{Bernhard Haeupler}, {and} \bibinfo{person}{Xiaorui Sun}.} \bibinfo{year}{2022}\natexlab{}.
\newblock \showarticletitle{Universally-Optimal Distributed Shortest Paths and Transshipment via Graph-Based $\ell$1-Oblivious Routing}. In \bibinfo{booktitle}{\emph{Annual ACM-SIAM Symposium on Discrete Algorithms (SODA)}}. \bibinfo{pages}{2549--2579}.
\newblock
\href{https://doi.org/10.1137/1.9781611977073.100}{doi:\nolinkurl{10.1137/1.9781611977073.100}}


\end{thebibliography}



\end{document}